\documentclass[notitlepage,superscriptaddress,tightenlines,preprintnumbers,twocolumn]{revtex4-1}
\usepackage{amsmath,verbatim,latexsym,amssymb,indentfirst,mathrsfs,mathtools,amsthm,bbm,bm,hyperref,url,cancel,enumitem,thm-restate}
\usepackage{verbatim,indentfirst}
\usepackage[title,titletoc]{appendix}
\usepackage{times}

\usepackage{graphicx}
\usepackage{epstopdf}

\renewcommand{\thesection}{\Roman{section}}
\renewcommand{\thesubsection}{\Roman{section} \Alph{subsection}}
\renewcommand{\thesubsubsection}{\Roman{section} \Alph{subsection} \arabic{subsubsection}}
\makeatletter
\def\p@subsection{}
\makeatother
\makeatletter
\def\p@subsubsection{}
\makeatother

\newtheorem{theorem}{Theorem}
\newtheorem{lemma}{Lemma}

\newtheorem{definition}{Definition}

\newtheorem{conjecture}{Conjecture}

\makeatletter
\newcommand\footnoteref[1]{\protected@xdef\@thefnmark{\ref{#1}}\@footnotemark}
\makeatother

\usepackage{soul}

\usepackage{color}
\usepackage[dvipsnames]{xcolor}
\usepackage[normalem]{ulem}

\usepackage{mathtools}
\DeclarePairedDelimiterXPP{\norm}[1]{}{\lVert}{\rVert}{}{{#1}}
\DeclarePairedDelimiterXPP{\opnorm}[1]{}{\lVert}{\rVert}{_\infty}{{#1}}
\DeclarePairedDelimiterXPP{\onenorm}[1]{}{\lVert}{\rVert}{_1}{{#1}}

\usepackage{phfcc}

\newcommand{\comp}{\mathcal{C}}
\newcommand{\tnorm}{\mathcal{T}}
\newcommand{\ent}[2]{H_{\rm c}^{#1, #2}}

\newcommand{\bw}{{\rm bw}}
\newcommand{\Qref}{Q_{\mathrm{ref}}}

\newcommand{\Max}{ {\rm max} }   
\newcommand{\hc}{ {\rm h.c.} }

\newcommand{\Tr}{{\rm Tr}}   
\def\id{\mathbbm{1}}   

\newcommand{\Sites}{n}  

\newcommand{\LParen}{ \bm{(} }
\newcommand{\RParen}{ \bm{)} }

\newcommand*{\Set}[1]{\left\{  #1  \right\}}

\newcommand*{\ket}[1]{\lvert {#1}\rangle}

\newcommand*{\ketbra}[2]{\lvert {#1}\rangle\!\langle {#2}\rvert}

\begin{document}
 
\title{Resource theory of quantum uncomplexity
}
\author{Nicole~Yunger~Halpern}
\email{nicoleyh@umd.edu}
\affiliation{Joint Center for Quantum Information and Computer Science, NIST and University of Maryland, College Park, MD 20742, USA}
\affiliation{Institute for Physical Science and Technology, University of Maryland, College Park, MD 20742, USA}
\affiliation{ITAMP, Harvard-Smithsonian Center for Astrophysics, Cambridge, MA 02138, USA}
\affiliation{Department of Physics, Harvard University, Cambridge, MA 02138, USA}
%
\author{Naga~B.\ T.\ Kothakonda}
\affiliation{Dahlem Center for Complex Quantum Systems, Freie Universit{\"a}t Berlin, 14195 Berlin, Germany}
\affiliation{Institute for Theoretical Physics, University of Cologne, D-50937 Cologne, Germany}

\author{Jonas~Haferkamp}
\affiliation{Dahlem Center for Complex Quantum Systems, Freie Universit{\"a}t Berlin, 14195 Berlin, Germany}
\affiliation{Helmholtz-Zentrum Berlin f{\"u}r Materialien und Energie, 14109 Berlin, Germany}
\author{Anthony~Munson}
\affiliation{Joint Center for Quantum Information and Computer Science, NIST and University of Maryland, College Park, MD 20742, USA}
\author{Jens~Eisert}
\affiliation{Dahlem Center for Complex Quantum Systems, Freie Universit{\"a}t Berlin, 14195 Berlin, Germany}
\affiliation{Helmholtz-Zentrum Berlin f{\"u}r Materialien und Energie, 14109 Berlin, Germany}
\author{Philippe~Faist}
\affiliation{Dahlem Center for Complex Quantum Systems, Freie Universit{\"a}t Berlin, 14195 Berlin, Germany}
\date{July 14, 2022} 

%
%
\begin{abstract}
Quantum complexity is emerging as a key property of many-body systems, including black holes, topological materials, and early quantum computers. A state's complexity quantifies the number of computational gates required to prepare the state from a simple tensor product. The greater a state's distance from maximal complexity, or ``uncomplexity,'' the more useful the state is as input to a quantum computation. Separately, resource theories---simple models for agents subject to constraints---are burgeoning in quantum information theory. We unite the two domains, confirming Brown and Susskind's conjecture that a resource theory of uncomplexity can be defined. The allowed operations, \emph{fuzzy operations,} are slightly random implementations of two-qubit gates chosen by an agent. We formalize two operational tasks, uncomplexity extraction and expenditure. Their optimal efficiencies depend on an entropy that we engineer to reflect complexity. We also present two monotones, uncomplexity measures that decline monotonically under fuzzy operations, in certain regimes. This work unleashes on many-body complexity the resource-theory toolkit from quantum information theory.
\end{abstract}
\maketitle  

Quantum complexity has recently swept from quantum computation
across many-body physics.
A state's \emph{quantum complexity} quantifies the difficulty of preparing the state from a simple fiducial state, often labeled $\ket{0^\Sites}$ for $\Sites$ qubits,
or of uncomputing a state to $\ket{0^\Sites}$.
For instance, a random circuit's output has a quantum complexity
that advantages certain quantum computations over classical competitors \cite{PhysRevLett.121.030501,Boixo,Arute_19_Quantum}.
Complexity also quantifies the difficulty of discrimination and of preparing superpositions~\cite{Aaronson_20_On,Girolami_21_Quantifying,Girolami_19_How}.
In condensed matter, topological phases are distinguished by complexities that scale linearly with the system size~\cite{Chen_11_Classification,Ahoronov_11_On,PhysRevLett.107.210501,PhysRevA.88.032321,Miller_18_Latent,Ali_20_Post,Liu_20_Circuit,Xiong_20_Nonanalyticity,Caputa_22_Quantum}.
In many-body physics, random 
evolutions increase complexity beyond when most physical quantities,
including correlators, equilibrate~\cite{Brandao_21_Models,Haferkamp_21_Linear,Li_22_Short,Brown_21_Quantum}.
Complexity saturation therefore forms a late stage of quantum many-body equilibration. 
This observation underpins a proposal about the
anti-de-Sitter-space/conformal-field-theory (AdS/CFT) holographic correspondence: There, a wormhole connecting two black holes is dual to a field-theoretic state. 
The state's complexity is conjectured to be proportional to
the wormhole's length~\cite{Susskind_16_Computational,Stanford_14_Complexity,Brown_16_Complexity,Brown_18_Second,Bouland_19_Computational,Brown_16_Holographic}.
Such diverse applications portray 
complexity as a physically impactful property.

Quantum computation is best begun with a low-complexity state:
A quantum computer needs ``clean'' qubits in the state $\ket{0^n}$
as we need blank paper when computing with a pencil.
\begin{figure}
  \centering
  \includegraphics{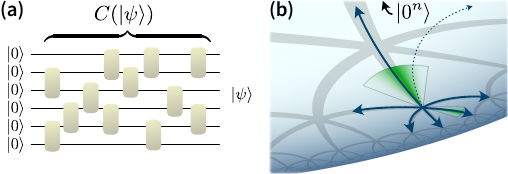}
  \caption{State complexity and its geometry.
    \textbf{(a)}~A pure $\Sites$-qubit state $\ket\psi$ has an \emph{exact circuit complexity} $\comp(\ket\psi)$ equal to the least number of gates required to prepare $\ket\psi$ from $\ket{0^\Sites}$.  The state's \emph{exact uncomplexity} is the distance
    $\comp_{\mathrm{max}} - \comp(\ket\psi)$ to 
    the maximal $\Sites$-qubit state complexity, $\comp_{\mathrm{max}} \sim e^n$.
    \textbf{(b)}~In a revision of Nielsen's
    geometry~\cite{Nielsen_05_Geometric,Susskind_18_Three},
    complexity curves the state space negatively~\cite{Brown_17_Quantum}.
    Applying a gate moves a state through 
    a unit length along some direction. 
    Most directions lead to higher-complexity states.
    To transform into a generic state $\ket{\psi'}$,
    $\ket\psi$ likely must pass through $\ket{0^\Sites}$;
    no significantly shorter path exists.
    Consequently, uncomplexity is desirable in quantum computation:
    Given a complex $\ket\psi$, to prepare a desired output $\ket{\psi'}$, one must uncompute $\ket{\psi'}$ to $\ket{0^\Sites}$.     Our resource theory slightly randomizes the gates (green shaded regions), leading them to realistically
    increase the state's complexity.}
  \label{fig:SusskindCircuitsAndComplexity}
\end{figure}
To quantify a state's resourcefulness in computation intuitively,
Brown and Susskind define a state's \emph{uncomplexity} as the gap between the state's greatest possible complexity, $\comp_\Max$, and actual complexity [Fig.~\ref{fig:SusskindCircuitsAndComplexity}\textbf{(a)}]~\cite{Brown_18_Second}. 
According to a counting argument, 
an $\Sites$-qubit state's $\comp_\Max$ scales as $e^\Sites$~\cite{Susskind_18_Three}. 
Uncomplexity appears to decrease monotonically under random 
dynamics;
complexity obeys an analog of the second law of thermodynamics~\cite{Brown_17_Quantum,Brown_18_Second,Susskind_18_Three,Bao_18_Quantum,Karar_18_Holographic,Lezgi_21_Complexity,Bai_22_Towards}.
This ``second law of uncomplexity'' led Brown and Susskind to conjecture that a resource theory for quantum uncomplexity can be defined. 
We formulate that resource theory precisely and use it to prove bounds on operational tasks' efficiencies.

A \emph{resource theory} is a simple quantum-information-theoretic model for an agent restricted to performing only certain operations. 
For example, in the resource theory of entanglement, agents can perform only local operations and classical communication. Resource theorists study
which state-to-state transformations the allowed operations can and cannot effect~\cite{Chitambar_19_Quantum}.
States impossible to prepare are scarce \emph{resources}, which may facilitate operational tasks. In the entanglement-theory example, entanglement is a resource usable to simulate a quantum channel~\cite{Horodecki_09_Quantum}.
If a resource theory's rules encode fundamental constraints of Nature, the conclusions extend from the agent's capabilities to natural evolutions.
Resource theories model diverse phenomena including informational nonequilibrium~\cite{Horodecki_03_Local,Horodecki_03_Reversible,Gour_15_Resource}, 
thermodynamics~\cite{Lieb_99_Physics,Janzing_00_Thermodynamic,Brandao_13_Resource,Faist_15_Gibbs,NYH_16_Beyond,NYH_18_Beyond,Lostaglio_17_Thermodynamic,Guryanova_16_Thermodynamics,NYH_16_Microcanonical,Sparaciari_17_Resource,Baghali_20_Resource},
coherence~\cite{Aberg_06_Quantifying,Baumgratz_14_Quantifying,Winter_16_Operational,Marvian_16_How,Dana_17_Resource,Streltsov_17_Towards,Bischof_19_Resource,Saxena_20_Dynamical,PhysRevLett.119.140402},
and quantum channels~\cite{Gour_19_Comparison,Takagi_19_General,Liu_19_Resource,Gour_19_How,Theurer_19_Quantifying,Liu_20_Operational}.
From their origins in quantum information theory,
resource theories have recently infiltrated other fields of science~\cite{NYH_17_Toward,Bernamonti_18_Holographic,Lorch_18_Optimal,Lostaglio_18_Elementary,Song_18_Quantifying,deLimaBernardo_18_Proposal,Alhambra_19_Heat,Chin_19_Partial,Holmes_19_Coherent,Henao_19_Experimental,NYH_20_Fundamental,Sarkar_20_Characterization,Messinger_20_Coherence,Tan_20_Identifying,Liu_20_Many,Luders_21_Quantifying,White_21_Conformal,Sparaciari_21_Bounding,Wu_21_Experimental,Paiva_21_Dynamics}.
Motivated by these studies' impact, 
we introduce a resource theory for quantum uncomplexity
at the intersection of high-energy theory and condensed matter.

We confirm Brown and Susskind's conjecture that
a resource theory for uncomplexity can be defined.
Upon defining the theory, we use it to
define two operational tasks:
\emph{Uncomplexity extraction} distills $\ket{0}$'s
from an arbitrary state.
\emph{Expending} uncomplexity, one can emulate an arbitrary
state. We quantify these tasks' optimal efficiencies with an entropy
that we introduce. This \emph{complexity entropy} measures the 
randomness that a state appears to have, to a realistic observer able to measure only simple observables.
In certain regimes, we prove, the complexity entropy is a \emph{monotone}, or resource measure, decreasing monotonically under allowed operations. 
This monotone result and another that we prove are
resource-theory versions of Brown and Susskind's second law of complexity.

A challenge in defining the resource theory follows from the agent's agency, or ability to choose operations. It is natural to formalize operations as gates for consistency with circuit-based complexity studies.  If able to implement any gates, though, the agent can uncompute any pure state to $\ket{0^n}$. 
Uncomplexity will 
not be a scarce resource; the resource theory will be a mockery. 
However, uncomputation circuits
lack robustness against imperfections in the gates' implementation.
If the gates are slightly noisy, a deep uncomputation circuit
likely prepares a highly mixed state, on average.
Yet mild noise should not significantly change qualitative outcomes
achievable by the agent; the outcomes should be robust.
We ensure this robustness by designating the allowed operations
as \emph{fuzzy gates}, slightly random approximations to the gates
that the agent wishes to perform,
modeling the noise in realistic circuit implementations.
Upon undergoing too many fuzzy gates, 
a state grows too random to be useful.
Hence the fuzziness prevents the agent from increasing a state's uncomplexity with high probability
[Fig.~\ref{fig:SusskindCircuitsAndComplexity}\textbf{(b)}].
This fuzziness also echoes the widespread modeling of chaos with randomness~\cite{Brown_11_Random,QuantumChaos,Hayden_07_Black,Cotler_17_Chaos,Bertini_20_Scrambling}.

This work is organized as follows.
We define the resource theory of uncomplexity,
then prove one resource-theory version of Brown and Susskind's second law of complexity. We formalize two operational tasks in the resource theory and quantify their efficiencies with a \emph{complexity entropy} that we define and that obeys another second law. 
We conclude with opportunities unveiled by this work.

\emph{Definition of the resource theory of uncomplexity.---}Consider a system of $\Sites$ qubits.
Denote by $\sigma_z$ the Pauli $z$-operator, 
by $\ket{0}$ its eigenvalue-1 eigenvector,
and by $\id$ the single-qubit identity operator.
Let $\ket{0^k}  :=  \ket{0}^{\otimes k}$.

Our resource theory's allowed operations consist of building blocks that we call \emph{fuzzy gates}. A fuzzy gate is effected when the resource-theory agent attempts to perform any desired gate 
$U \in \mathrm{SU}(4)$ on any two qubits.
(Our results extend to alternative gate sets favored in the holographic literature~\footnote{
First, each gate can couple $k \geq 2$ qubits.
Second, the target gates $U$ can form a discrete set.}\cite{Brown_18_Second,Lin_19_Cayley}.)
The implemented gate is a slightly random variation on the target gate, as motivated in the introduction~\footnote{Usually, to make a resource theory robust against noise, one only needs for the agent to tolerate a nonzero probability of preparing the target state incorrectly. In complexity-related contexts, though, individual gates' errors can snowball throughout a circuit, rather than appearing only in the end result. We therefore incorporate the errors into the allowed operations}.
We model the randomness as follows.
Denote by $d \tilde{U}$ the Haar (uniform) measure over SU(4).
Fix an error parameter $\epsilon>0$.
Denote by $p_{U,\epsilon}(\tilde{U})$ any normalized probability density over the two-qubit gates $\tilde{U}$ that satisfies two assumptions:
(i)~$p_{U,\epsilon}$ introduces noise in all directions of the two-qubit--gate space around $U$, being nonzero on an open set that contains $U$. 
(ii)~The measure $p_{U,\epsilon}(\tilde{U})$ vanishes for all unitaries $\tilde{U}$ far from the target gate: $\opnorm{U - \tilde{U}}>\epsilon$, wherein $\norm{.}_\infty$ denotes the operator norm. 
The implemented gate is a $\tilde{U}$ chosen according to 
the measure $p_{U,\epsilon}(\tilde{U})\,d \tilde{U}$ \footnote{
We can illustrate the probability density $p_{U,\epsilon}$ with two examples.
First, $\tilde{U}$ can be chosen uniformly randomly from the two-qubit unitaries $O(\epsilon)$-close to $U$ in any norm.
Second, denote by $\{ P_{j,k} \}$ any basis for
the traceless 2-qubit Hermitian operators.
Assign random coefficients 
$\alpha_{j,k} \in [ -O(\epsilon), \, O(\epsilon) ]$.
The Hamiltonian
\unexpanded{$H = \sum \alpha_{j,k} P_{j,k}$}
perturbs $U$ into $\tilde{U} = e^{iH} U$.
}.

\begin{definition}[Fuzzy gates and operations]
\label{def_Fuzzy_Op}
Denote by $U$ an arbitrary two-qubit gate.
The \emph{fuzzy gate} $\tilde{U}$ is selected randomly
according to any distribution 
$p_{U,\epsilon}(\tilde{U}) \, d\tilde{U}$ that satisfies conditions~(i) and~(ii) above.
Every composition of fuzzy gates is a \emph{fuzzy operation}.
\end{definition}

A resource theory's allowed operations
form a set closed under composition~\cite{Chitambar_19_Quantum}.
We therefore choose our free operations to be fuzzy operations, which include all compositions of fuzzy gates; the fuzzy operations form a set closed under composition.
The fuzziness suggests two variants of the resource theory.
In one variant, the initial state is pure, and the agent knows which
gates are applied. A unitary models the evolution, and the state remains pure. 
Holographic literature motivates this variant~\cite{Brown_11_Random,QuantumChaos,Hayden_07_Black,Cotler_17_Chaos,Bertini_20_Scrambling},
whereas quantum-information conventions motivate the second.
In the second variant, whenever applying a fuzzy gate, the agent lacks any knowledge of the noise sample. All possible instances of the gate are averaged over, increasing the state's mixedness~\cite{Horodecki_03_Local,Horodecki_03_Reversible,Gour_15_Resource,Chiribella_17_Microcanonical}.
We further motivate and analyze both variants below.

The allowed operations exclude 
the tensoring on and discarding of states.
This lack, although unusual, has precedents~\cite{Sparaciari_17_Resource,Chiribella_17_Microcanonical}. 
No states are free because any tensored-on state benefits quantum computation:
Consider tensoring a maximally complex $m$-qubit state
onto a maximally complex $n$-qubit state.
$\comp_\Max$ grows to $\sim e^{n+m}$,
whereas the actual complexity grows only to $\sim e^n + e^m$~\footnote{The number of gates needed to prepare the tensor product equals the number needed to prepare one factor plus the number needed to prepare the other factor.}
Hence the tensoring-on raises the uncomplexity from 0.
Even tensoring on a maximally mixed state can boost computational power,
as shown by the one-clean-qubit computational
model~\cite{Knill_98_Power,Brown_18_Second,Susskind_18_Three}.
The allowed operations exclude the discarding of subsystems because the resource theory is intended to model Brown and Susskind's setup~\cite{Brown_18_Second}---a system whose Hilbert space remains fixed.

\emph{Second law of complexity for pure-state variant.---}In the resource theory's first variant,
the initial state is pure, the agent always knows which gate is applied, and the evolution is unitary.  
This setting is
common in condensed-matter theory and high-energy physics. 
There, a unitary circuit of randomly sampled gates mimics chaos in certain ways~\cite{Brown_11_Random,QuantumChaos,Hayden_07_Black,Cotler_17_Chaos,Bertini_20_Scrambling}.
This scenario precludes challenges such as defining mixed-state complexity.
In this variant, we prove a version of Brown and Susskind's ``second law of complexity''~\cite{Brown_18_Second}---in resource-theory parlance, a \emph{monotone} statement. Monotones are functions $f$ that quantify a resource's monotonic decline under allowed operations.
For any state $\rho$ and allowed operation $\mathcal{E}$,
$f (\rho)  \geq  f  \LParen \mathcal{E}(\rho)  \RParen$.
Different monotones quantify a state's usefulness in different tasks.
For example, consider extracting work by thermalizing an arbitrary state $\rho$ (analogously to extracting work from an expanding gas) or performing work to prepare $\rho$ 
(analogously to compressing a gas).
The extractable and required work are monotones in a thermodynamic resource theory~\cite{Horodecki_13_Fundamental}. 
Monotones resemble free energy, but each resource theory has multiple monotones; there is no ``one monotone to rule them all''~\cite{Gour_15_Resource}.

We prove that two functions are 
fuzzy-operation monotones in certain regimes.
The conditionality reflects the notorious difficulty of proving
that complexity measures grow monotonically under random dynamics~\cite{Knill_95_Approximation,Nielsen_05_Geometric,Nielsen_06_Quantum,Nielsen_06_Optimal,Gosset_14_Algorithm,Roberts_17_Chaos,Brown_18_Second,Balasubramanian_20_Quantum,Brandao_21_Models,Eisert_21_Entangling,Haferkamp_21_Linear,Balasubramanian_21_Complexity}.
The first monotone depends on a \emph{brickwork circuit}, a common circuit formed from staggered layers of gates (Fig.~\ref{fig_Brickwork} in App.~\ref{app_Prove_Monotones}).
Define the \emph{brickwork complexity} $\mathcal{C}_\bw(|\psi\rangle)$ as the least number of gates in any brickwork circuit that prepares a pure state $|\psi\rangle$.
The \emph{brickwork uncomplexity} is 
$\comp_\Max - \comp_\bw (\ket{\psi})$.

\begin{theorem}[``Second law'' for brickwork uncomplexity: informal]
\label{thm:monotone_informal}
Let $\ket{\psi}$ denote an arbitrary $\Sites$-qubit pure state.
The brickwork uncomplexity $\comp_{\mathrm{max}} - \comp_\bw(\ket{\psi})$
cannot increase under any fuzzy brickwork circuit $\tilde{U}$
of $\geq n$ layers, except in a measure-0 set of $\ket{\psi}$-preparation-and-$\tilde{U}$-sampling experiments.
\end{theorem}
\noindent
We prove a more technical version of the theorem in App.~\ref{app_Prove_Monotones}.
The proof extends random-circuit results in Ref.~\cite{Haferkamp_21_Linear},
leveraging assumption~(i) in Definition~\ref{def_Fuzzy_Op}.
This second law is for quantum complexity, not for an entropy of the state averaged over noise samples. Such entropies were shown to obey second laws in Ref.~\cite{Gour_15_Resource}. We contrast these entropies with complexity below.
Also, we prove another second law for complexity (another monotone) in our resource theory's second variant.
Both second laws hold even if the fuzziness $\epsilon$ is arbitrarily small---even constant in $\Sites$, such that fuzzy gates mimic the target (ideal) evolution with constant-in-$\Sites$ precision for a time.

\emph{Resource-theory variant~2: Mixed-state evolution.---}In the remainder of this paper, the resource-theory agent does not know which noise sample is realized during any fuzzy gate.
All noise instances are effectively averaged over;
a fuzzy gate implements the quantum channel
$\mathcal{E}(.) = \int \tilde{U}(.)\tilde{U}^\dagger
\,p_{U,\epsilon}(\tilde{U})\,d\tilde{U}$.
The corresponding fuzzy operations form a strict subset of
the set of noisy operations, allowed in the resource theory of informational nonequilibrium~\cite{Horodecki_03_Local,Horodecki_03_Reversible,Gour_15_Resource,Chiribella_17_Microcanonical}.
All results proved about variant~2 of our resource theory
are true also of variant~1.

Variant~2 offers a new approach to defining mixed-state complexity, following Brown
and Susskind's resource-centric vision~\cite{Brown_18_Second}.
A common notion of mixed-state complexity quantifies the 
gates required to prepare a purification of $\rho$
from $\ket{0^{2\Sites}}$.  This measure is the \emph{purification complexity}~\cite{Caceres_20_Complexity,Ruan_20_Purification,Camargo_21_Entanglement},
which differs from the notion of mixed-state complexity captured by our resource theory.
We illustrate with the $\Sites$-qubit maximally mixed state, $\id^{\otimes \Sites} / 2^\Sites$.
A purification of $\id^{\otimes \Sites} / 2^\Sites$ consists of $\Sites$ Bell pairs (maximally entangled states~\cite{NielsenC10}).
Each pair results from starting with $\ket{0^2}$, then performing one single-qubit rotation and one two-qubit entangling gate.
Hence the purification complexity of $\id^{\otimes \Sites} / 2^\Sites$ is upper bounded by $2\Sites$. 

However, $\id^{\otimes \Sites} / 2^\Sites$ 
is invariant under every unitary and so is useless for quantum computation, in the absence of additional qubits.
So is a highly complex state: 
Starting with such a complex state, 
one needs many gates to uncompute even a few qubits to $\ket{0}$'s.
If the agent cannot perform so many gates,
a complex state benefits quantum computation as little as a maximally mixed state does~\cite{Brown_18_Second} (App.~\ref{appx:MixedStateComplexity}).
Our resource theory correspondingly casts a state's complexity as 
the difficulty of extracting $\ket{0}$'s from the state.
See App.~\ref{appx:MixedStateComplexity} for further elaboration.
We quantify the difficulty of extracting $\ket0$'s with a new entropic quantity.

\emph{Complexity entropy.---}We introduce an entropy that quantifies tasks' efficiencies in the resource theory of uncomplexity.
Common entropies do not reflect complexity~\cite{Susskind2014arXiv_notenough}.
For instance, consider a chaotic system evolving unitarily from $\ket{0^\Sites}$. 
The state's von Neumann entropy remains constant, even as the state grows highly complex.
Furthermore, a small subsystem's reduced state tends to equilibrate
on short time scales, so entanglement entropies saturate quickly.
In contrast, the complexity can grow for a time 
$\sim e^\Sites$~\cite{Haferkamp_21_Linear}.
Failing to encode the complexity's time scales,
ordinary entropies cannot capture complexity.
We overcome this obstacle, introducing an entropy that quantifies complexity, inspired by Ref.~\cite{Brandao_21_Models}.
Reference~\cite{Kothakonda_21_Entropy} will detail the entropy's properties. 
References~\cite{Hastad1999SJC_pseudorandom,Chen2017arXiv_computational} introduced related quantities, motivated by 
pseudorandomness and cryptography.

The complexity entropy quantifies how random a state appears if probed only through simple observables.
For instance, consider measuring a simple observable of a highly complex state $\ket{\psi}$.
The outcome is highly random, as if $\ket{\psi}$ were highly entropic~\cite{Gross2009PRL_most}.
Reference~\cite{Brandao_21_Models} introduces 
a strategy for quantifying this apparent randomness:
Quantify the state's distinguishability from the maximally mixed state in an operational task implementable with a limited number of steps.
Inspired by this approach, we use the \emph{hypothesis-testing
entropy} as our quantifier~\cite{Hiai_91_Proper,Dupuis_13_Generalized,HypothesisEntropy,Hayashi,WatrousScripts,Wang2012PRL_oneshot,Tomamichel2013_hierarchy}.

In a hypothesis test, one receives a state, $\rho$ or $\sigma$, and guesses which state arrived. The most general strategy involves a two-outcome measurement, represented quantum-information-theoretically with a positive-operator-valued measure (POVM)~\cite{NielsenC10}
$\{ Q, \id^{\otimes \Sites}  - Q \}$.
Each measurement operator is positive-semidefinite:
$0\leq Q \leq \id^{\otimes \Sites} $.
Outcome $Q$ suggests that the state was $\rho$, 
and $\id^{\otimes \Sites} - Q$ suggests that the state was $\sigma$.
Let $\sigma = \id^{\otimes \Sites} /2^\Sites$.
Suppose that one must, if the state is $\rho$,
guess $\rho$ with a probability at least $\eta \in (0, 1]$.
The minimum probability of wrongly guessing $\rho$,
if the state is $\sigma = \id^{\otimes \Sites} /2^\Sites$, 
defines the hypothesis-testing entropy~\cite{HypothesisEntropy,Wang2012PRL_oneshot,Dupuis2013_DH},
\begin{align}
    H_{\mathrm{h}}^\eta(\rho)
    := \log_2 \Bigg( \min_{\substack{0 \leq Q \leq \id \\ \Tr(Q\rho)\geq\eta}} 
    \bigl\{ \Tr(Q) \bigr\}  \Bigg) \ .
    \label{eq:defn-hypothesis-testing-entropy}
\end{align}

We restrict the measurement's computational difficulty: 
First, we define a set $M_0$ of
zero-complexity measurement operators.
Under such an operator's action, each qubit is 
(i)~projected onto $\ket{0}$ or 
(ii)~not touched (evolved with $\id$).
Define the variable $\alpha_j$ as 1 if qubit $j$ is projected
and as 0 otherwise. If $(\ketbra{0}{0})^0 \equiv \id$,
the measurement operator has the form
\begin{align}
   \label{eq_M_r_Elt_simple}
   \bigotimes_{j=1}^\Sites
   ( {_j}\ketbra{0}{0}_j )^{\alpha_j} 
   = Q_0 \in M_0 \ .
\end{align}
To progress beyond zero-complexity measurements, we
fix an integer $r \geq 0$.
Consider performing $\leq r$ two-local gates,
effecting a unitary $U_r$, before measuring a $Q_0$ \footnote{
The resource theory inspires this limitation on the number of gates performable. We are currently defining the complexity entropy as a quantum-information-theoretic quantity not hitched to the resource theory. However, the resource theory inspires this entropy's definition, and the entropy will be applied in the resource theory. The resource-theory agent may have a tolerance for how fuzzy a state may grow and so may choose to restrict how many fuzzy gates they perform. This restriction inspires the complexity entropy's $r$.}.
The net effect, we define as a \emph{complexity-$r$ measurement}.
The operators $U_r^\dag Q_0 U_r$ form a set $M_r$.
Restricting to $M_r$ the $Q$ in~\eqref{eq:defn-hypothesis-testing-entropy},
we define the complexity entropy.

\begin{definition}[Complexity entropy]
\label{def_Comp_Ent}
The \emph{complexity entropy} of any 
state $\rho$ is, for any error tolerance $\eta \in (0, 1]$,
\begin{align}
   \label{eq_Def_Complex_Ent}
   \ent{r}{\eta} ( \rho )
   & := \min_{\substack{ Q \in M_r, \\ \Tr( Q \rho ) \geq \eta} }
   \{ \log_2  \LParen  \Tr (Q)  \RParen  \} .
\end{align}
\end{definition}
We can understand the definition through two extremes,
detailed in App.~\ref{app_Entropy_Extremes}
and synopsized here.
Suppose that $\rho = \ketbra{\psi}{\psi}$ is pure.
First, let the number $r$ of performable gates be at least 
the number of gates needed to prepare $\ket{\psi}$---let $\ket{\psi}$ be relatively uncomplex.
The complexity entropy will attain its minimum at
$\ent{r}{\eta} (\ket{\psi}) = 0$.
Contrariwise, let $\ket{\psi}$ be highly complex 
and only a few gates be performable (let $r$ be small). 
The complexity entropy maximizes: 
$\ent{r}{\eta} ( \ket{\psi} ) = \Sites$. 

We can choose different conventions in Definition~\ref{def_Comp_Ent}, to suit different setups.
First, we can define $M_0$ in terms of the measurements natural for a given platform.
Second, we can define $M_r$ in terms of any complexity measure, such as Nielsen's~\cite{Nielsen_05_Geometric,Nielsen_06_Quantum,Nielsen_06_Optimal,Dowling_06_Geometry}, rather than in terms of $r$ gates.

The complexity entropy features in our second monotone 
(``second law'' for complexity).
The \emph{complexity negentropy}
$\Sites - \ent{r}{\eta} ( \rho )$
quantifies how far from maximally mixed $\rho$ looks under limited-complexity measurements.
The complexity negentropy declines monotonically in two cases,
which involve two definitions:
(i) Define an \emph{architecture} as the layout of gates in a quantum circuit.
(ii) Define as $\mathcal{E}_{k,r}$ the ensemble of 
$\Sites$-qubit states formed as follows:
Pick $k = 0, 1, \ldots, \Sites$ qubits uniformly randomly;
and pick a $k$-qubit state vector $\ket{\phi}$ Haar-randomly.
Prepare those qubits in $\ket{\phi}$.
Pad $\ket{\phi}$ with $\ket{0}$'s, to produce 
$\ket{\phi}  \ket{ 0^{\Sites - k} }$.
Perform a circuit, with a random architecture, of any $\leq r$ fuzzy gates.

\begin{restatable}[First case of the complexity negentropy's monotonicity]{lemma}{monotoneentropyI}
\label{lemma:monotoneentropyI}
Consider drawing a state from $\mathcal{E}_{k,r}$ uniformly randomly,
then performing an arbitrary fuzzy gate.
Let the number of chosen qubits be $k > \log_2 (15 r)$.
The complexity negentropy $\Sites - \ent{r}{1}$
does not increase, with probability 1.
\end{restatable}
\noindent Appendix~\ref{app_H_Monoton} presents the proof, as well as 
the second (more involved) case of
the complexity negentropy's monotonicity. 

In both cases, $\eta = 1$, enabling us to apply
the algebraic-geometry toolkit of Ref.~\cite{Haferkamp_21_Linear}.
If the accuracy threshold $\eta \ll 1$, a fuzzy gate can decrease $\ent{r}{\eta}$, as $\eta$ can absorb sufficiently small fuzziness.
We therefore conjecture the complexity negentropy's monotonicity whenever $\eta$ is sufficiently large.

\begin{conjecture}[Monotonicity of the complexity negentropy]
\label{conj_H_Mon}
Let $\rho$ denote any $\Sites$-qubit quantum state.
The complexity negentropy $\Sites  -  \ent{r}{\eta} (\rho)$ 
cannot increase under fuzzy operations, 
for any $r \in \mathbb{Z}_{\geq 0}$, 
if $\eta \geq \eta_0(\epsilon)$, 
for some function $\eta_0(\epsilon)$.
\end{conjecture}

\emph{Uncomplexity extraction and expenditure.---}%
The complexity entropy, we prove, quantifies the optimal efficiencies of two operational tasks that we formalize, using the resource theory (Fig.~\ref{fig:UncomplexityExtractionExpenditure}).
\begin{figure}
  \centering
  \includegraphics{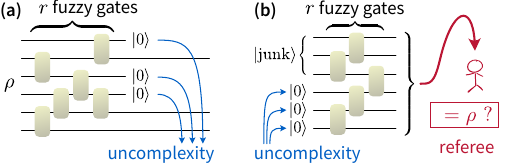}
  \caption{Operational tasks of uncomplexity extraction and
    expenditure in the resource theory.  
    Since gates are fuzzy, the agent can perform only $\leq r$ gates, lest the state grow too noisy to be useful.
    \textbf{(a)}~Extracting uncomplexity from a state $\rho$,
    one applies $\leq r$ fuzzy gates.
    The number of qubits left in the state $\ket{0}$ is 
    the extractable uncomplexity, which equals the complexity entropy of $\rho$.
    \textbf{(b)}~Given enough $\ket{0}$'s,
    an agent can ``spend'' uncomplexity to imitate $\rho$:
    The agent performs $r$ fuzzy gates, preparing a state believed, by a computationally bounded referee, to be $\rho$.}
  \label{fig:UncomplexityExtractionExpenditure}
\end{figure}
The allowed operations' fuzziness, recall, limits the number of gates performable before the state grows too random to be useful.
Our theorems therefore concern an agent limited to performing $\leq r$ gates. 
We quantify states' closeness with the trace distance, 
$\tnorm(\rho, \bar{\rho}) := \onenorm{\rho-\bar{\rho} }/2$, 
wherein $\onenorm{A} = \Tr\sqrt{A^\dagger A}$ 
denotes the trace norm.

We define \emph{uncomplexity extraction} as follows.
Let $\rho$ denote any $\Sites$-qubit state.
For any tolerance $\delta \geq 0$,
we seek a circuit of at most $r \in \mathbb{Z}_{\geq 0}$ fuzzy gates,
and a selection of $w$ qubits, with the following property.
Suppose that $\rho$ undergoes the circuit and then
the nonselected qubits are discarded.
The result is $\delta$-close to $\ket{0^{w}}$ in trace distance
[Fig.~\ref{fig:UncomplexityExtractionExpenditure}\textbf{(a)}].
The following theorem establishes an extraction protocol's existence and near-optimality.
Appendix~\ref{app_Prove_Extract} contains the proof.

\begin{restatable}[Uncomplexity extraction]{theorem}{thmUncompExt}
\label{thm_Uncomp_Ext}
Let $\rho$, $r$, and $\delta$ be as above, and
assume that $\delta \geq r \epsilon$.
For every parameter value
$\eta \in [1 - (\delta-r\epsilon)^2, \, 1]$,
some protocol extracts
$w = \Sites - H_{\mathrm{c}}^{r,\eta}(\rho)$
qubits $\delta$-close to $\ket{0^w}$ in trace distance.
Conversely, every uncomplexity-extraction protocol obeys
$w \leq \Sites - H_{\mathrm{c}}^{r,1-\delta}(\rho)$.
\end{restatable}
\noindent 
Theorem~\ref{thm_Uncomp_Ext} endows the complexity entropy with an operational significance in a quantum-computational task, beyond quantifying states' indistinguishability.

Another bound governs the uncomplexity cost of mimicking a state $\rho$.
Suppose that a computationally limited referee,
upon receiving an $\Sites$-qubit state, 
performs a hypothesis test between
$\rho$ and the completely ignorant observer's null hypothesis, 
$\id^{\otimes \Sites} / 2^\Sites$.
Computationally restricted, the referee can measure only 
operators $Q \in M_r$.
Given $\rho$, the referee must guess $\rho$
with a probability $\geq \eta \in (0, 1]$.
Naturally, the referee minimizes the probability of
guessing $\rho$ when given $\id^{\otimes \Sites} / 2^\Sites$.
Knowing the referee's choice of $Q$
(many $Q$'s can be optimal~\cite{Kothakonda_21_Entropy}),
the agent tricks the referee by preparing 
a simulacrum $\tilde{\rho}$.
The agent borrows $w \leq \Sites$ uncomplex $\ket{0}$'s 
from an ``uncomplexity bank'' (e.g., someone else's laboratory).
The bank tacks on whichever $(\Sites - w)$-qubit state $\sigma$ 
is handy, to raise the total number of qubits to $\Sites$.
The agent knows which qubits are $\ket{0}$'s but not $\sigma$'s form.
The agent transforms the $\Sites$-qubit state with $\leq r$ gates.
The referee must guess, 
with a probability $\geq 1 - \delta \in (0, 1]$, 
that the output is $\rho$.

\begin{restatable}[Uncomplexity expenditure]{theorem}{thmUncompCost}
\label{thm_Uncomp_Cost}
Let $\rho$ denote any $\Sites$-qubit state.
Let $r$ and $\delta$ be as above, 
and assume that the error tolerance is $\delta \geq 2r \epsilon$.
For every $\eta \in (0, 1]$,
and for every $(\Sites-w)$-qubit state $\sigma$,
$\rho$ can be imitated with 
$w = \Sites - \ent{r}{\eta}(\rho)$ uncomplex $\ket{0}$'s.
\end{restatable}
\noindent
Appendix~\ref{app_Prove_Expend} contains the proof.
We expect that 
$\Sites - \ent{r}{\eta}(\rho)$ uncomplex $\ket{0}$'s are necessary.

\emph{Conclusions.---}We have confirmed Brown and Susskind's conjecture~\cite{Brown_18_Second} that a resource theory of uncomplexity can be defined.
The resource theory's allowed operations balance random evolutions, which model features of chaos,
with the agency in resource theories---the agent chooses operations to perform.
We proved two variations on Brown and Susskind's second law of complexity---resource-theory monotones.
Using the resource theory, we formalized uncomplexity expenditure and extraction.
The tasks' optimal efficiencies, we quantified with a complexity entropy that we introduced.
This work introduces into quantum complexity
the resource-theory toolbox that has garnered successes across quantum information theory~\cite{Gour_15_Resource}.

Our resource theory deviates superficially from two holographic conventions.  We invoke the circuit complexity, instead of Nielsen's geometric distance; correspondingly, gates act in discrete time steps, whereas Hamiltonians act continuously.  Our model is motivated by (i)~quantum information theory, where discrete gates form circuits; (ii)~the closeness of circuit complexity to Nielsen's complexity~\cite{Nielsen_06_Quantum}; and
(iii)~random circuits' exhibiting features of chaos.
Reasons (ii) and (iii) suggest that our results might
extend from fuzzy gates to perturbed continuous-time evolutions.

This work establishes several opportunities for future research.
First, the complexity entropy can quantify the efficiencies of tasks other than those defined here. 
Examples include randomness extraction under computational restrictions~\cite{Kothakonda_21_Entropy}.
The complexity entropy may be typically difficult to compute 
but easier to bound. Experimental strategies include hypothesis testing and, combined with Theorem~\ref{thm_Uncomp_Ext}, uncomplexity extraction.
Reference~\cite{Kothakonda_21_Entropy} will explore the complexity entropy's properties and applications.

Second, the complexity entropy suggests an operational answer to a question of active research: how to define mixed-state complexity~\cite{Agon_19_Subsystem,Caceres_20_Complexity,Ruan_20_Purification,Camargo_21_Entanglement,Ruan_21_Circuit,Saha_21_Holographic,Brandao_21_Models}. 
According to our results, complexity quantifies the difficulty of extracting uncomplex $\ket{0}$ qubits. More precisely, the complexity entropy can anchor a version of the strong complexity introduced in Ref.~\cite{Brandao_21_Models}.

Third, proving Conjecture~\ref{conj_H_Mon} would cement the complexity negentropy's interpretation as a resource quantifier. 
Also, a proof would elevate the converse in Theorem~\ref{thm_Uncomp_Ext} to governing arbitrarily many fuzzy gates:
One would evaluate the complexity entropy on $\rho$
and on the post-circuit state, 
then invoke the entropy's monotonicity.

Fourth, one can try to prove that the resource theory of uncomplexity has (or lacks) properties common to resource theories~\cite{Chitambar_19_Quantum}.
For example, can allowed operations interconvert any two states
asymptotically (if arbitrarily many copies are available)?
Furthermore, we anticipate connections with the resource theory of
magic, another model for the difficulty of implementing unitaries~\cite{Veitch2014NJP_resource,Veitch2012NJP_negative,Mari2012PRL_positive,Howard_17_Application,Ahmadi_18_Quantification}.

Fifth, the resource theory can impact holography, many-body physics,
and quantum computation. One might reframe black-hole paradoxes in terms of uncomplexity extraction and expenditure, then prove quantitative results using the resource theory. For example, an agent falling into a black hole wishes to remove firewalls from the horizon, to avoid burning. Tossing in a thermal photon doubles the time for which the agent remains safe~\cite{Susskind_16_Typical}. How much time does a given state---so a given amount of uncomplexity---buy?
Also, uncomplexity's monotonicity under fuzzy circuits is expected to relate to the \emph{switchback effect}~\cite{Stanford_14_Complexity,Brown_18_Second},
which determines how perturbations affect complexity's evolution.
The present work, providing a quantitative
resource theory of uncomplexity as a technical tool, is hoped to galvanize
further studies of space-time's uncomplexity.

\emph{Acknowledgements.---}The authors thank Adam Brown, Giulio Chiribella, Eric Chitambar, Gilad Gour, Aram Harrow, Richard K\"ung, Jonathan Oppenheim, Carlo Maria Scandolo, and Leonard Susskind for valuable discussions.
N.~Y.~H.~thanks John Preskill for a nudge toward the problem of defining a resource theory of uncomplexity and thanks Shira Chapman, Michael Walter, and the other organizers of the 2020 Lorentz Center complexity workshop for inspiration.
This research was supported by NSF grants for the Institute for Theoretical Atomic, Molecular, and Optical Physics at Harvard University and the Smithsonian Astrophysical Observatory; 
grant no.~NSF PHY-1748958; 
and QLCI grant OMA-2120757. 
A.M. was supported by NIST grant 70NANB21H055\_0. 
The Berlin team was funded by the FQXi, the DFG (FOR 2724 and CRC 183), and the Einstein Foundation.

\onecolumngrid

\begin{appendices}

\renewcommand{\thesection}{\Alph{section}}
\renewcommand{\thesubsection}{\Alph{section} \arabic{subsection}}
\renewcommand{\thesubsubsection}{\Alph{section} \arabic{subsection} \roman{subsubsection}}

\makeatletter\@addtoreset{equation}{section}
\def\theequation{\thesection\arabic{equation}}

\section{Proof of the brickwork complexity's monotonicity}
\label{app_Prove_Monotones}

This appendix contains the proof of Theorem~\ref{thm_Monotones}.
We extend results in Ref.~\cite{Haferkamp_21_Linear}
from Haar-random gates to fuzzy gates, then to monotonicity statements. 
Several pieces of background are necessary.
We call an arrangement of gates an \emph{architecture}.
Slotting particular gates into an architecture
produces a circuit.
A circuit may contain a \emph{light cone}, a block of gates
that contains one qubit that links to each other qubit
via a path, perhaps unique to the latter qubit, formed from gates (Fig.~2 of Ref.~\cite{Haferkamp_21_Linear}).

An \emph{accessible dimension} is introduced in Ref.~\cite{Haferkamp_21_Linear}:
Consider choosing two-qubit gates whose fuzzy approximations 
are slotted into any architecture $A$.
The slotting-in forms a circuit that implements a unitary.
All the unitaries so implementable form a set $\mathcal{U}(A)$.
The number of degrees of freedom needed to specify $\mathcal{U}(A)$
is the \emph{accessible dimension} 
$\dim \LParen \mathcal{U}(A) \RParen  
\leq  4^\Sites$~\cite{Haferkamp_21_Linear}.
Consider applying to $\ket{0^\Sites}$ each unitary in $\mathcal{U}(A)$. 
The resulting states form a set that we denote by 
$\mathcal{U}_{\rm state} (A)$.

A fuzzy gate, recall, is drawn from the group SU(4) of two-qubit gates.
Denote by $d \tilde{U}$ the Haar measure on SU(4).
By definition, fuzzy gates are drawn according to probability distributions of the form 
$p_{U, \epsilon} (\tilde{U}) \, d \tilde{U}$, 
for a density function $p_{U, \epsilon} (\tilde{U})$ in $L^1$.
Therefore, the fuzzy-gate distribution is \textit{absolutely continuous} with respect to the Haar measure:
Consider any set of gates that has measure $0$ with respect to the Haar measure on $\mathrm{SU}(4)^{\times R}$. 
The set has measure 0 also for the fuzzy-gate distribution.

By an {$\Sites$-qubit Pauli string}, we mean,
an $\Sites$-fold tensor product of 
(i) one-qubit Pauli operators and 
(ii) one-qubit identity operators $\id$.
We prove the following technical version of Theorem~\ref{thm:monotone_informal},
which governs brickwork circuits 
(Fig.~\ref{fig_Brickwork}).

\begin{figure}[hbt]
\centering
\includegraphics[width=.25\textwidth, clip=true]{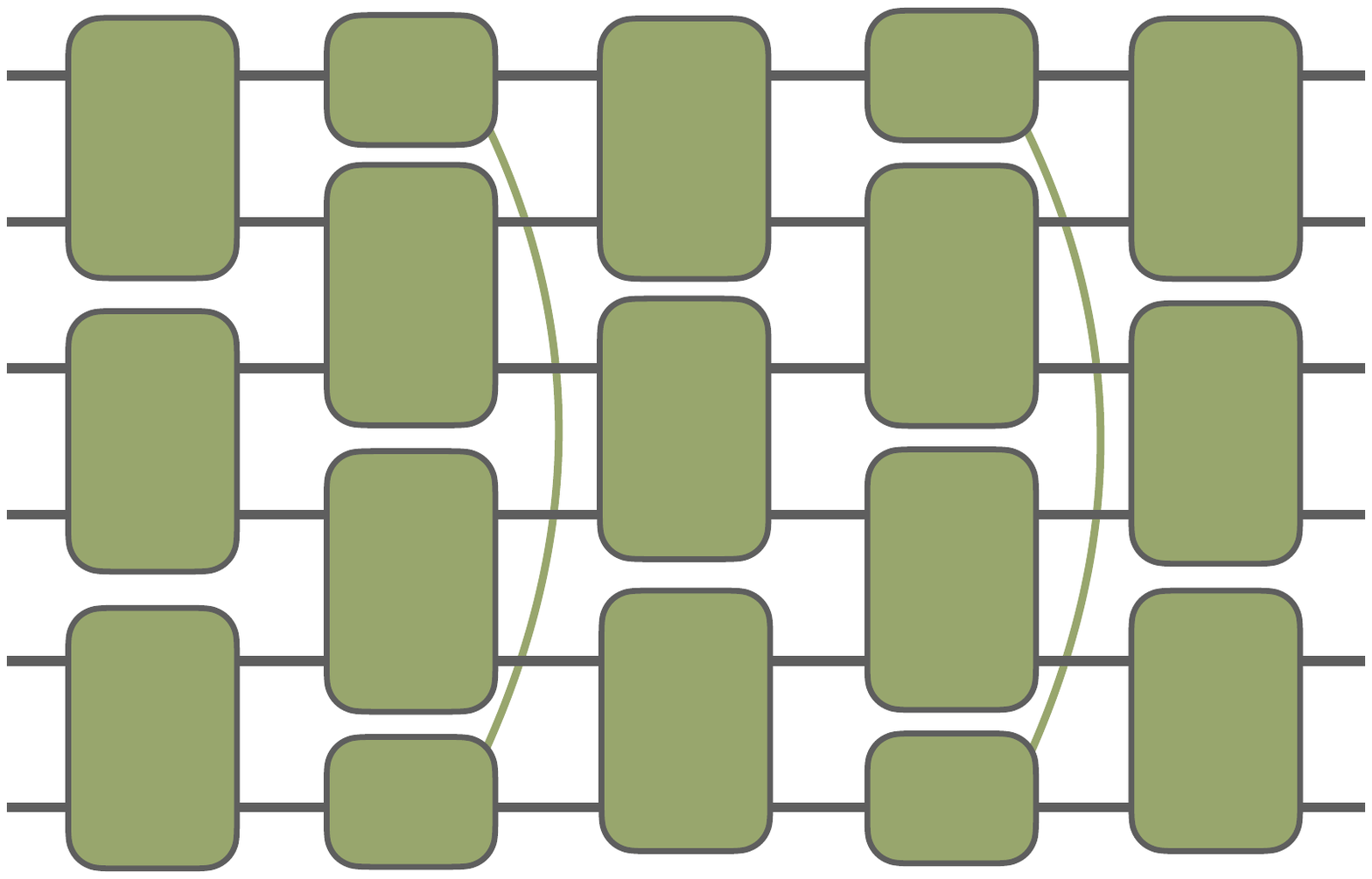}
\caption{A brickwork circuit of the sort featured in Theorem~\ref{thm:monotone_informal}.
Periodic boundary conditions impose a gate (green boxes) on qubits $\Sites$ and 1 in each even-indexed layer.}
\label{fig_Brickwork}
\end{figure}

\begin{theorem}[Monotonicity of exact uncomplexity: formal]
\label{thm_Monotones}
Consider the following protocol:
Choose any two-qubit gates $U_1, U_2, \ldots, U_R$,
for some $R \in \mathbb{Z}_+$. 
Let $\epsilon'>0$, and
draw $\tilde{U}_1, \tilde{U}_2, \ldots, \tilde{U}_R$
near $U_1, U_2, \ldots, U_R$
according to an $\epsilon'$-fuzzy distribution over $\mathrm{SU}(4)^{\times R}$.
Slot the resulting gates into any brickwork architecture.
Apply the resulting circuit to $\ket{0^\Sites}$.
This protocol prepares a state whose brickwork uncomplexity is
$\geq  \comp_\Max  -  R$.
For every integer $R \in \LParen 0,  \,  \Omega (4^\Sites)  \RParen$, the following holds:
Let $V_1, V_2, \ldots, V_{\Sites(\Sites-1)}$  denote any two-qubit gates in an $\Sites$-layer brickwork architecture. 
Let $\tilde{V}_1, \tilde{V}_2, \ldots, \tilde{V}_{\Sites(\Sites-1)}$ denote
corresponding $\epsilon$-fuzzy gates, with $\epsilon > 0$.
With probability 1, the brickwork uncomplexity decreases:
\begin{align}
   \label{eq_Monoton_Particular_Arch} 
   \comp_\Max - \comp_\bw \left( 
   \left[ \tilde{V}_{\Sites(\Sites-1)} 
   \tilde{V}_{\Sites(\Sites-1)-1} \ldots \tilde{V}_1 \right] 
   \left[ \tilde{U}_R \tilde{U}_{R-1} \ldots \tilde{U}_1 \right] \ket{0^\Sites} \right)
   <  \comp_\Max  -  R \ .
\end{align}
\end{theorem}

$\epsilon$ and $\epsilon'$ can be chosen to be arbitrarily small without affecting the theorem.  
To connect the formal statement above with the main text's informal
statement (Theorem~\ref{thm:monotone_informal}), we proceed as follows.
We choose for the unitaries $U_1, U_2, \ldots, U_R$ to form an optimal brickwork preparation circuit for $\ket\psi$. 
We choose for $\epsilon'$ to be much smaller than
other parameters in the problem.
The latter choice ensures that the state transformed by $V_1, V_2, \ldots,V_{n(n-1)}$ is arbitrarily close to $\ket\psi$.

We need the following lemma, proven as Lemma~1 in Ref.~\cite{Haferkamp_21_Linear}:
	\begin{lemma}\label{lemma:subsetlemma}
	Let $A$ denote any architecture.
	The states preparable with architecture-$A$ circuits form the set $\mathcal{U}_{\mathrm{state}}(A)$.
	Let $M\subset \mathcal{U}_{\mathrm{state}}(A)$ denote any (semialgebraic) subset for which
	$\dim (M) < \dim \LParen \mathcal{U}_{\mathrm{state}}(A) \RParen$.
	Consider drawing an architecture-$A$ circuit uniformly randomly.
	The circuit effects a unitary in $M$ with probability 0.
	\end{lemma}
	
\begin{proof}[Proof of Theorem~\ref{thm_Monotones}]
To prove that the brickwork uncomplexity decreases monotonically, we prove that the brickwork complexity increases monotonically.
Denote by $A_{\bw,T}$ the $T$-layer brickwork architecture, which contains $T(\Sites-1) = R$ gates total.
To prove Ineq.~\eqref{eq_Monoton_Particular_Arch}, we must prove only that
\begin{align}
   \label{eq_Reduce_Proof}
   \dim \LParen 
   \mathcal{U}_{\mathrm{state}}( A_{\bw, T+\Sites} ) \RParen 
   > \dim \LParen \mathcal{U}_{\mathrm{state}}(A_{\bw,T}) \RParen.
\end{align}
The reason is, Ineq.~\eqref{eq_Reduce_Proof} and Lemma~\ref{lemma:subsetlemma} imply the following: 
Consider randomly drawing a unitary effected by a brickwork architecture $A_{\bw,T+\Sites}$.
The unitary has zero probability of being implementable with a brickwork architecture whose $R= T(\Sites-1)$.
		
Let us prove Ineq.~\eqref{eq_Reduce_Proof}.
Consider contracting the gates in $A_{\bw,T}$, 
then applying the resulting unitary to $\ket{0^\Sites}$.
We map a set of $R$ two-qubit gates to 
a point on the unit sphere.
More generally, contraction forms a smooth map
$f : \mathrm{SU}(4)^{\times R} \to S^{ 2\times 2^\Sites-1}$.
The map's greatest possible rank equals 
$\dim \LParen \mathcal{U}_{\mathrm{state}}(A) \RParen$,
we show via semialgebraic geometry in Ref.~\cite{Haferkamp_21_Linear}.
The map's rank also---by definition---equals the rank of the map's Jacobian. 
Therefore, to prove Ineq.~\eqref{eq_Reduce_Proof}, 
we must identify one circuit---one point 
$x=(U_1, U_2, \ldots, U_R, U_{R+1}, \ldots, U_{R+\Sites(\Sites-1)}) 
\in \mathrm{SU}(4)^{\times [R+\Sites(\Sites-1)]}$---for which 
the map's Jacobian has a rank 
$> \dim \LParen \mathcal{U}_{\mathrm{state}}(A_{\bw,T}) \RParen$.
	
We construct that circuit as follows.
Let $P_j$ denote a two-local Pauli operator that acts
on the same qubits as $U_j$.
The Jacobian's image is spanned by 
$\{ (U'_R U'_{R-1} \ldots U'_{j+1}) P_j 
(U'_{j} U'_{j-1} \ldots U'_1) \}_{j, P}$~\cite{Haferkamp_21_Linear}.
Denote by $x_{\mathrm{max}} \in \mathrm{SU}(4)^{\times R}$ 
a point at which the $A$ contraction map achieves its maximal rank,
$r_{\mathrm{max}}$.
We will construct a point 
$x_{\mathrm{ext}} = (U'_1, U'_2, \ldots, U'_R, \:
V'_1, V'_2, \ldots, V'_{\Sites(\Sites-1)})$ 
at which the contraction map's rank exceeds the rank at $r_{\mathrm{max}}$.
(The $V'_j$ operators are gates chosen to prove a variation on Theorem~\ref{thm_Monotones}. In the variation, Haar-random gates replace the fuzzy gates $\tilde{V}_j$. 
By absolute continuity, the theorem follows for fuzzy gates.)

There exist Hermitian operators $H_j$ such that the following is true: The map has a Jacobian whose image
at $x_{\mathrm{max}}$ is spanned by
$\{H_j U_R U_{R-1} \ldots U_1|0^\Sites \rangle\}_{j = 1,2, \ldots r_{\mathrm{max}} } =: \{|v_j\rangle\}$.
These vectors' span excludes some of the states formed by applying
an $\Sites$-qubit Pauli string to
$U_R U_{R-1} \ldots U_1|0^{\Sites}\rangle$~\footnote{
We can prove this claim by contradiction: Assume that, for all Pauli operators $P$,
$P U_r U_{r-1} \ldots U_1 | 0^{\Sites} \rangle \in \mathrm{span} \{ | v_j \rangle \}$.
The Pauli operators form an orthonormal basis for the Hermitian operators defined on the same Hilbert space. 
Therefore, for every state $|\psi\rangle$ in the space, we can build the Hermitian operator
$A' = | \psi \rangle \langle 0^{\Sites} |
(U_1^{' \dagger} U_2^{'\dagger} \ldots U_r^{'\dagger} )
+ \hc$
By the assumption we mean to contradict, 
$A' (U'_r U'_{r-1} \ldots U'_1) |0^{\Sites}\rangle
=|\psi\rangle
\in \mathrm{span}\{|v_j\rangle\}$.
However, we can derive a contradiction to the foregoing equation.
According to the text above, the Jacobian's image has a submaximal rank. Therefore, $\mathrm{span}\{|v_j\rangle\}$ is not the entire space. 
Therefore, we can choose for $|\psi\rangle$ to be orthogonal to the $|v_j\rangle$.
However, we have already proved that 
$\ket{\psi} \in \mathrm{span}\{|v_j\rangle\}$.
We have proven a contradiction, so our first premise is false.}.
The gates $V'_1, V'_2, \ldots, V'_{\Sites(\Sites-1)}$ transform some excluded Pauli string $P$ into $Z_\ell$, for a to-be-specified qubit $\ell$:
$(V'_{\Sites(\Sites-1)} V'_{\Sites(\Sites-1)-1} \ldots V'_1) \, P \, 
\left( V_1^{' \dagger} V_2^{'\dagger} \ldots  
V_{\Sites(\Sites-1)}^{' \dagger} \right)
= Z_\ell$. 
The foregoing claim is true, Ref.~\cite{Haferkamp_21_Linear} shows, under two necessary conditions:
(i) The gates $V'_1, V'_2, \ldots, V'_{\Sites(\Sites-1)}$ are in an architecture that contains a light cone.
(ii) $\ell$ indexes the light cone's final qubit, which connects to all the other light-cone qubits via paths formed from gates.
Both conditions are satisfied by an
$\Sites(\Sites-1)$-gate brickwork circuit and $\ell=1$.
All the states 
$\left( V'_{\Sites(\Sites-1)} V'_{\Sites(\Sites-1)-1} 
\ldots V'_1 \right) H_j 
\left( V_1^{' \dagger} V_2^{'\dagger} \ldots V_{\Sites(\Sites-1)}^{'\dagger} \right)
\left( U'_R U'_{R-1} \ldots U'_1 \right)|0^{\Sites}\rangle$ and 
$Z_\ell (U'_R U'_{R-1} \ldots U'_1)|0^{\Sites}\rangle$ 
(i) are in the Jacobian's image and 
(ii) are linearly independent by assumption. 
Therefore, the contraction map's Jacobian has a rank
$\dim \LParen \mathcal{U}_{\mathrm{state}} (A_{\bw, T}) \RParen$.

Therefore, Ineq.~\eqref{eq_Reduce_Proof} holds for circuits extended with Haar-random gates $V'_j$.
The fuzzy-gate probability distribution is absolutely continuous
with respect to the Haar measure, as explained in the beginning of this appendix. Therefore,  Theorem~\ref{thm_Monotones} is true if fuzzy gates replace the Haar-random gates. 
\end{proof}

\section{Mixed-state complexity and the complexity entropy}
\label{appx:MixedStateComplexity}
\label{app_Entropy_Extremes}

This appendix elaborates on quantifiers of mixed-state complexity. We situate the complexity entropy~\eqref{eq_Def_Complex_Ent} in the landscape of such quantifiers and provide intuition about the entropy. 
Throughout this appendix, $\rho$ denotes an arbitrary $\Sites$-qubit mixed state.

Following Brown and Susskind, we quantify complexity in terms of a state's usefulness for quantum computation. We therefore seek a notion of complexity that assigns to $\id^{\otimes \Sites} / 2^\Sites$ a complexity similar to a maximally complex pure state's.  Our resource theory is designed to reflect such a notion of complexity.
Choosing fuzzy operations to be free ensures that every state can be mapped to the maximally mixed state for free, in variant 2 of the resource theory. The ordering of states in our resource theory is therefore compatible with assigning a high complexity to the maximally mixed state.

Our resource theory's notion of complexity is further motivated by the physical justification for fuzzy operations.
Many resource theories, such as the commonest resource theory of
thermodynamics~\cite{Brandao_13_Resource}, have the following property.  
Suppose that the agent performs a free operation that is $\epsilon$-close to the desired operation, perhaps because noise corrupts the implementation. 
The applied operation prepares the desired output state to within some
error bounded by $\epsilon$.  
Therefore, statements made using this resource theory 
are robust with respect to small errors in an operation's implementation.
This property is essential, lending a resource theory its physical,
operational significance.
In contrast, suppose that a resource theory predicted that a transformation $\rho\mapsto\sigma$ were possible via some operation,
yet a slightly different operation would produce a state far from $\sigma$.
The resource theory would lose its operational significance.

Requiring similar robustness of the uncomplexity resource theory leads to the choice of fuzzy operations as the free operations.
The free operations are defined in terms of elementary operations that can be composed. 
If the agent wishes to implement some unitary $U$, 
they must decompose $U$ into elementary operations.
Each elementary operation might 
suffer from some imperfection in any practical setting.  
Therefore, to render the resource theory robust with respect to implementation errors, we must ensure that physical statements about
the resource theory are robust with respect to perturbations affecting the desired elementary operations. 
A natural way of enforcing this property is to explicitly
model the noise affecting the implemented gates. 
This strategy results in fuzzy operations' being the free operations. 

Fuzzy operations suggest the complexity entropy as 
an alternative approach to defining mixed-state complexity,  
quantifying a state's usefulness in quantum computation.  
Instead of defining one complexity measure for $\rho$,
we define a family of measures parameterized by 
$r \in \mathbb{Z}_{\geq 0}$.
For any fixed $r$, we ask how many $\ket{0}$ qubits
can be extracted from $\rho$, with a probability $\geq \eta$,
via an application of $\leq r$ fuzzy gates.
This quantity essentially equals the complexity entropy $H_{h}^{r,\eta}(\rho)$,
according to Theorem~\ref{thm_Uncomp_Ext}.

We could invert this relation to solve for 
the minimal number $r$ of gates required to extract 
a fixed number $k$ of $\ket0$ qubits.  
However, two problematic situations might arise.  
First, $\rho$ might be fundamentally mixed:
Extracting a large number $k$ of $\ket0$ qubits might be impossible, even with arbitrarily many gates and with an arbitrarily small fuzziness parameter $\epsilon$. 
Second, even if $\rho$ is low-rank, 
the number $r$ of perfect gates needed to extract $k$ pure $\ket{0}$ qubits may be too large to be implementable with fuzzy gates, if $\epsilon$ is too large.
In each of these two situations, the number $r$ of gates required to extract $k$ pure $\ket{0}$ qubits is ill-defined.
Therefore, this number---the inverted relation between $r$ and $k$---forms an incomplete measure of mixed-state complexity. 
The complexity entropy offers a well-defined alternative.

We now expound upon the motivation for the complexity entropy's definition.
The definition is inspired by the
\emph{strong complexity} of Brand\~ao \emph{et al.}~\cite{Brandao_21_Models}.
They quantify the distinguishability of a state $\rho$ from the maximally mixed state
using a measurement whose complexity is limited.  More precisely, they quantify
the probability of successfully distinguishing $\rho$ from $\id^{\otimes n}/2^n$ when each state is provided with a probability $1/2$. Their distinguishability
measure serves as an extension of the trace distance, which quantifies the optimal efficiency with which states can be distinguished through arbitrarily complex measurements.
We consider an alternative distinguishability setting, in which the prior probabilities of receiving $\rho$ or $\id^{\otimes n}/2^n$ are unknown.  If the measurement can be arbitrarily complex, the probability of successfully identifying $\id^{\otimes n}/2^n$, while successfully identifying $\rho$ with probability at least $\eta$, is given by the
hypothesis-testing entropy of $\rho$.  Our complexity entropy is hence an extension of the hypothesis-testing entropy.
Therefore, the complexity entropy importantly has an operational significance beyond the strong complexity of Ref.~\cite{Brandao_21_Models}:
the number of $\ket{0}$ qubits extractable from $\rho$ with a bounded number of limited-precision gates (Theorem~\ref{thm_Uncomp_Ext}).
The complexity entropy's operational significance is underscored by its having units---namely, bits. 

The complexity entropy [Def.~\ref{def_Comp_Ent}]
can be understood through two limits.
Synopsized in the main text, they are detailed here.
In both cases, we suppose that 
$\rho = \ketbra{\psi}{\psi}$ is pure, for simplicity.

First, let the number $r$ of performable gates be at least 
the number of gates needed to prepare $\ket{\psi}$.
$U_r$ undoes the gates in $\ket{\psi}$.
The tensor factors in~\eqref{eq_M_r_Elt_simple}
can therefore be $\ketbra{0}{0}$'s:
Hence $Q = U_r \ketbra{0^\Sites}{0^\Sites} U_r^\dag$
can project onto one low-dimensional subspace,
such that $\Tr(Q) = 1$ is small, as the minimization requires.
If the projected-onto subspace is so small, is
$\Tr (Q \ketbra{\psi}{\psi} ) \geq \eta$ violated?
No, as $U_r \ket{\psi} = \ket{0^\Sites}$ by assumption. 
Therefore, if $\ket{\psi}$ is uncomplex
while many gates are available,
$\ent{r}{\eta} (\ket{\psi}) = 0$.

Contrariwise, let $\ket{\psi}$ be highly complex 
and only a few gates be performable (let $r$ be small). 
Most qubits, probed locally, likely resemble $\id / 2$.
Each $\id / 2$ halves $\Tr (Q \ketbra{\psi}{\psi})$ 
if multiplying a $\ketbra{0}{0}$ in $Q$.
To satisfy $\Tr(Q  \ketbra{\psi}{\psi}) \geq \eta$,
$Q$ must contain many $\id$'s.
Each $\id$ doubles $\Tr(Q)$.
In the extreme case, $Q = \id^{\otimes \Sites}$, 
and $\ent{r}{\eta} ( \ket{\psi} )  =  \Sites$.
Therefore, if $\ket{\psi}$ is complex
while few gates are performable, 
$\ent{r}{\eta} ( \ket{\psi} ) \lesssim \Sites$.


We conclude this appendix by commenting on the disconnect between the notion of preparation complexity and the notion of distinguishability from the maximally mixed state. The preparation complexity quantifies the elementary operations required to prepare a state, while the distinguishability quantifies the complexity of a measurement required to tell a state from $\id^{\otimes n}/2^n$.
Even for pure states, the notions differ widely~\cite{Brandao_21_Models}.
Consider $\ket0\otimes\ket{\psi_{n-1}}$, wherein $\ket{\psi_{n-1}}$ is highly complex.
This state's circuit complexity is $\sim 2^{n}$.  Yet measuring the first qubit in the computational basis distinguishes the state from the maximally mixed state with high probability. Hence this state's strong complexity vanishes.
For all pure states $\ket{\psi}$, the preparation complexity exceeds the complexity of distinguishing $\ket{\psi}$ from $\id^{\otimes n}/2^n$.  Indeed, a preparation circuit for $\ket\psi$ automatically provides
a scheme for distinguishing $\ket\psi$ from $\id^{\otimes n}/2^n$~\cite{Brandao_21_Models}.
For mixed states $\rho$, however, the preparation complexity can lie below the complexity of distinguishing $\rho$ from $\id^{\otimes n}/2^n$.
(This discussion might depend on the specifics of how either notion of complexity is defined.  For concreteness, we regard the purification complexity as the preparation complexity of $\rho$. We also regard the complexity entropy as quantifying the distinguishability of $\rho$ from $\id^{\otimes n}/2^n$.)
$\id^{\otimes n}/2^n$ is easy to prepare, provided we have $n$ pure ancillary qubits.  However $\id^{\otimes n} / 2^n$ is indistinguishable from itself; in the language of the complexity
entropy, extracting $\ket0$ qubits from $\id^n/2^n$ is impossible, regardless of our computational power.
This discrepancy may result partially from the computational model used to define the purification complexity, which involves pure ancillary qubits. Overall, we see no physical reason to worry about the existence of two distinct notions
of complexity. They quantify different physical properties of a state---even pure states.

\section{Uncomplexity-extraction proof}
\label{app_Prove_Extract}

First, we prove a lemma used in the proofs of Theorems~\ref{thm_Uncomp_Ext} and~\ref{thm_Uncomp_Cost}:
Suppose that an arbitrary state $\omega$ undergoes
a desired $r$-gate unitary $U_r$ or
a fuzzy approximation $\tilde{U}_r$.
The transformed states, $U_r \omega U_r^\dag$ and
$\tilde{U}_r \omega \tilde{U}_r^\dag$,
are $r\epsilon$-close in trace distance.
Then, we complete the proof of Theorem~\ref{thm_Uncomp_Ext}.

\begin{lemma}
\label{lemma_Triangle_Inequality}

Let $\omega$ denote an arbitrary $\Sites$-qubit state. 
Consider transforming $\omega$ with 
perfectly implemented gates $V_1, V_2, \ldots, V_r$. 
Each $V_j$ is defined on $\mathbb{C}^{2\Sites}$ 
but transforms just one qubit subspace nontrivially. 
The gates effect the unitary
$U_r := V_r V_{r-1} \ldots V_1$ and yield the state
$U_r \omega U_r^\dag \ .$ 
Suppose that the gates implemented are $\epsilon$-fuzzy. 
Analogously, denote the fuzzy gates by
$\tilde{V}_1, \tilde{V}_2, \ldots, \tilde{V}_r$
and the effected unitary by
$\tilde{U}_r := \tilde{V}_r \tilde{V}_{r-1} \ldots \tilde{V}_1$. 
The fuzzy gates yield the state 
$\tilde{U}_r \omega \tilde{U}_r^\dag \ .$ 
The transformed states are $r\epsilon$-close in trace distance:
\begin{align}
   \label{eq_TriangleLemma_Help1}
   \tnorm \left( U_r \omega U_r^\dag, \, \tilde{U}_r \omega \tilde{U}_r^\dag \right)
   \leq r \epsilon \ .
\end{align}
\end{lemma}

\begin{proof}
We prove the lemma in two steps. First, consider transforming an arbitrary $\Sites$-qubit state $\tau$ with a perfectly implemented gate $V_j$ or a fuzzy gate $\tilde{V}_j$. The fuzzy gates are $\epsilon$-close to the perfect gates in operator norm,
consistently with Definition~\ref{def_Fuzzy_Op}:
\begin{align}
   \label{eq_Triangle_Help1}
   \opnorm{ \tilde{V}_j - V_j }
   \leq  \epsilon 
   \quad  \forall j = 1, 2, \ldots, r \ .
\end{align}
Therefore, the transformed states are close
in the sense that
\begin{align}
   \tnorm \left( V_j \tau V_j^\dag, \,
   \tilde{V}_j \tau \tilde{V}_j^\dag  \right)
   \label{eq_Triangle_Help2}
   &\leq  \tnorm \left( 
   \tilde{V}_j \tau \tilde{V}_j^\dagger,   \, 
    V_j \tau \tilde{V}_j^\dagger \right)
   +  \tnorm \left( V_j \tau \tilde{V}_j^\dagger,  \,
                     V_j \tau V_j^\dagger  \right) \\
   \label{eq_Triangle_Help3}
   &=  \frac{1}{2} \,  \Bigl\lVert
   \left( \tilde{V}_j-V_j \right)  \tau  \tilde{V}_j^\dagger 
   \Bigr\rVert_1
   +  \frac{1}{2} \, \Bigl\lVert
   V_j  \tau \left( \tilde{V}_j - V_j\right)^\dagger 
   \Bigr\rVert_1  \\
   \label{eq_Triangle_Help4}
   &= \frac{1}{2} \, \Bigl\lVert \left( \tilde{V}_j-V_j \right)  \tau \Bigr\rVert_1
        +  \frac{1}{2} \, \Bigl\lVert
             \tau \left( \tilde{V}_j-V_j \right)^\dagger \Bigr\rVert_1 \\
   \label{eq_Triangle_Help5}
   &= \bigl\lVert  (\tilde{V}_j-V_j)  \tau  \bigr\rVert_1 \\
   \label{eq_Triangle_Help6}
   & \leq  \opnorm[\big]{ \tilde{V}_j  -  V_j }  \;
   \onenorm{\tau} \\
   \label{eq_Triangle_Help7}
   & \leq \epsilon .
\end{align}
Inequality~\eqref{eq_Triangle_Help2} follows from the triangle inequality.
Equation~\eqref{eq_Triangle_Help4} follows from the trace distance's unitary invariance. 
Equation~\eqref{eq_Triangle_Help5} follows from the property
$\onenorm{ A } = \onenorm{ A^\dag }$
of all linear operators $A$. 
Inequality~\eqref{eq_Triangle_Help6} follows from H\"older's inequality.
Inequality~\eqref{eq_Triangle_Help7} follows from~\eqref{eq_Triangle_Help1} 
and 
from $\onenorm{ \tau } = 1$.

Second, we prove Ineq.~\eqref{eq_TriangleLemma_Help1} inductively.
The bound~\eqref{eq_TriangleLemma_Help1} holds when $r=1$:
\begin{align}
   \tnorm \left( U_r \omega U_r^\dag, \, \tilde{U}_r \omega \tilde{U}_r^\dag \right)
   = \tnorm \left( 
   V_1 \omega V_1^\dag ,  \,
   \tilde{V}_1 \omega \tilde{V}_1^\dag  \right)
   \leq \epsilon \ .
\end{align}
Suppose that Ineq.~\eqref{eq_Triangle_Help2} holds
for all $r' = 1, 2, \ldots, r - 1$.
Define $U_{r-1} := V_{r-1} V_{r-2} \ldots V_1$ and
$\tilde{U}_{r-1} := \tilde{V}_{r-1} \tilde{V}_{r-2}
\ldots \tilde{V}_1$.
We can use these unitaries to bound the original trace distance as
\begin{align}
   \tnorm \left( U_r \omega U_r^\dag , \,
   \tilde{U}_r \omega \tilde{U}_r^\dag \right)
   & \leq \tnorm \left( V_r U_{r-1} \omega U_{r-1}^\dag V_r^\dag, \,
   V_r \tilde{U}_{r-1} \omega \tilde{U}_{r-1}^\dag V_r^\dag \right) 
   \nonumber \\ & \qquad
   \label{eq_Triangle_Help8}
   + \tnorm \left( V_r \tilde{U}_{r-1} \omega \tilde{U}_{r-1}^\dag V_r^\dag , \,
   \tilde{V}_r \tilde{U}_{r-1} \omega \tilde{U}_{r-1}^\dag \tilde{V}_r^\dag \right) \\
   \label{eq_Triangle_Help9}
   & \leq (r - 1) \epsilon + \epsilon \\
   & = r \epsilon \ .
\end{align}
Inequality~\eqref{eq_Triangle_Help8} follows from the triangle inequality.
Inequality~\eqref{eq_Triangle_Help9} follows from (i) the trace distance's unitary invariance,
(ii) the inductive hypothesis, and (iii) the application of Ineq.~\eqref{eq_Triangle_Help7} to
$\tau = \tilde{U}_{r-1} \omega \tilde{U}_{r-1}^\dag$.
\end{proof}

Having proved a lemma used in Theorems~\ref{thm_Uncomp_Ext} and~\ref{thm_Uncomp_Cost}, we complete the proof of Theorem~\ref{thm_Uncomp_Ext}.

\thmUncompExt*

\begin{proof}
First, we prove that extracting $w = \Sites - \ent{r}{\eta}(\rho)$ uncomplex $\ket{0}$'s from $\rho$ is achievable.
Then, we show that this number is optimal:
No protocol can extract more $\ket{0}$'s.

\emph{Achievability:}
Let $\rho$ denote any $\Sites$-qubit state, 
$\eta \in (0, 1]$, and $r \in \mathbb{Z}_{\geq 0}$.
Consider any $Q$ that achieves 
the minimization in Eq.~\eqref{eq_Def_Complex_Ent}.
$Q$ projects $\Sites - \ent{r}{\eta} (\rho)$ qubits onto $\ket{0}$.
Without loss of generality, we index those qubits as 
$1, 2, \ldots, \Sites - \ent{r}{\eta}(\rho)$.
Denote that set of qubits by $W$;
and the rest of the qubits, by $\bar{W}$.
Denote by $U_r$ the $(\leq r)$-qubit unitary used to implement $Q$:
$Q = U_r^\dag 
\left( \ketbra{0^{\Sites - \ent{r}{\eta}(\rho)}}{0^{\Sites - \ent{r}{\eta}(\rho)}} \otimes \id^{\otimes \ent{r}{\eta}(\rho)}
\right) U_r$.
By the constraint in the definition~\eqref{def_Comp_Ent},
\begin{align}
   \eta
   & \leq \Tr ( Q \rho )
   = \Tr \left( 
   U_r^\dag 
   \left[ \ketbra{0^{\Sites - \ent{r}{\eta}(\rho)}}{0^{\Sites - \ent{r}{\eta}(\rho)}} \otimes \id^{\otimes \ent{r}{\eta}(\rho)}
   \right] U_r
   \rho \right) \\
   \label{eq_Ext_Achiev_Help1}
   & = \Tr_{ \bar{W} } \left(
   \Tr_W (U_r \rho U_r^\dag)
   \ketbra{0^{\Sites - \ent{r}{\eta}(\rho)}}{
           0^{\Sites - \ent{r}{\eta}(\rho)}}
   \right) \ .
\end{align}
The trace's cyclicality implies the final equality.
Equation~\eqref{eq_Ext_Achiev_Help1} implies that
$\Tr_W (U_r \rho U_r^\dag)$ has a fidelity $\geq \eta$
to $\ketbra{0^{\Sites - \ent{r}{\eta}(\rho)}}{
            0^{\Sites - \ent{r}{\eta}(\rho)}}$.
By the relationship between the fidelity and the trace distance~\cite[Theorem~9.3.1]{Wilde_11_From},
\begin{align}
   \label{eq_Ext_Achiev_Help2}
   \tnorm \LParen 
   \Tr_W (U_r \rho U_r^\dag) , \,
   \ketbra{0^{\Sites - \ent{r}{\eta}(\rho)}}{
           0^{\Sites - \ent{r}{\eta}(\rho)}} \RParen
   \leq \sqrt{1 - \eta} \ .
\end{align}
Therefore, if $U_r$ is implemented perfectly---if the gates are implemented perfectly---the protocol extracts $w = \Sites - \ent{r}{\eta}(\rho)$
uncomplex $\ket{0}$'s with accuracy
$\leq \sqrt{1 - \eta} \ .$

Now, suppose that the gates are $\epsilon$-fuzzy.
Attempting to implement $U_r$, the agent actually implements an approximation $\tilde{U}_r$.
By Lemma~\ref{lemma_Triangle_Inequality},
$\tnorm \left( U_r \rho U_r^\dag, \,
\tilde{U}_r \rho \tilde{U}_r^\dag \right)
\leq r \epsilon \ .$
The trace distance is contractive under all completely positive, trace-preserving maps, including the partial trace. Therefore,
$\tnorm \LParen
\Tr_W (U_r \rho U_r^\dag), \,
\Tr_W (\tilde{U}_r \rho \tilde{U}_r^\dag) \RParen
\leq r \epsilon \ .$
We combine this inequality with Ineq.~\eqref{eq_Ext_Achiev_Help2},
using the triangle inequality:
\begin{align}
   \tnorm \LParen
   \Tr_W ( \tilde{U}_r  \rho  \tilde{U}_r^\dag ) , \,
   \ketbra{0^{\Sites - \ent{r}{\eta}(\rho)}}{
           0^{\Sites - \ent{r}{\eta}(\rho)}} \RParen
   & \leq \tnorm \LParen
   \Tr_W ( \tilde{U}_r \rho \tilde{U}_r^\dag ), \,
   \Tr_W ( U_r \rho U_r^\dag )
   \RParen
   \\ & \qquad \nonumber
   + \tnorm \LParen
   \Tr_W ( U_r \rho U_r^\dag ) , \,
   \ketbra{0^{\Sites - \ent{r}{\eta}(\rho)}}{
           0^{\Sites - \ent{r}{\eta}(\rho)}}
   \RParen \\
   & \leq \sqrt{1 - \eta} + r \epsilon \ .
\end{align}
Therefore, for $\delta \geq \sqrt{1 - \eta} + r \epsilon$,
or $\eta \geq 1 - (\delta -r \epsilon)^2$,
 one can extract $w = \Sites - \ent{r}{\eta}(\rho)$
uncomplex $\ket{0}$'s, with accuracy
$\geq \delta$, using $\epsilon$-fuzzy operations.

\emph{Optimality:}
Again, denote by $W$ the set of not-discarded qubits,
and index them as the first qubits.
Denote by $\bar{W}$ the set of discarded qubits.
By the constraints on the extraction protocol,
the final state must be $\delta$-close to the $\ketbra{0^w}{0^w}$:
If $\tilde{U}_r$ denotes the circuit performed with $\leq r$ fuzzy gates,
\begin{align}
   \label{eq_Eta_close}
   \tnorm \LParen
   \Tr_{ \bar{W} }( \tilde{U}_r \rho \tilde{U}_r^\dag ), 
   \ketbra{0^{ w}}{0^{w}}\RParen
   \leq \delta.
\end{align}
We can recast this result in terms of hypothesis testing, 
using the following quantum-information result:
Let $\sigma$ and $\gamma$ denote quantum states 
defined on the same Hilbert space, 
and let $\Lambda$ denote any operator such that $0 \leq \Lambda \leq \id$.
According to Corollary~9.1.1 of Ref.~\cite{Wilde_11_From},
\begin{align}
   \label{eq_Wilde_Ineq}
   \tnorm (\sigma, \gamma)
   \geq  \Tr (\Lambda \gamma)  -  \Tr (\Lambda \sigma) .
\end{align}
Let $\sigma = \Tr_{ \bar{W} } 
\left( \tilde{U}_r \rho \tilde{U}_r^\dag \right)$, and let
$\gamma = \Lambda = \ketbra{0^{w} }{0^{w} }$.
We substitute into Ineq.~\eqref{eq_Wilde_Ineq},
combine the result with Ineq.~\eqref{eq_Eta_close}, 
and rearrange terms. The result is
\begin{align}
   \label{eq_ExtractThm_Help2}
   \Tr \LParen \ketbra{0^{w} }{0^{w} }
   \Tr_{ \bar{W} } (\tilde{U}_r \rho \tilde{U}_r^\dag)  \RParen
   \geq1- \delta .
\end{align}
Let us rewrite the trace's argument such that
each factor is defined on the $\Sites$-qubit Hilbert space. 
Padding the outer product with identity operators at the discarded sites yields
\begin{align}
   \label{eq_Def_Q_0}
   \ketbra{0^w}{0^w}
   \otimes  \id^{\otimes (\Sites - w)}
   =: Q_0 .
\end{align}
Therefore, by Eq.~\eqref{eq_ExtractThm_Help2}, 
$\Tr \left( Q_0 
\left[\tilde{U}_r \rho \tilde{U}_r^\dag \right] \right)  
\geq  1-\delta$.
Let us cycle the $\tilde{U}_r^\dag$ leftward.
Packaging up $\tilde{U}_r^\dag  Q_0  \tilde{U}_r  =:  \bar{Q}$ implies
$\Tr ( \bar{Q} \rho )
   \geq 1-\delta $.
By the foregoing inequality, and by 
the number of the fuzzy gates that constitute $U_r$,
$\bar{Q}$ is in $M_r$. 
By the minimum in Eq.~\eqref{eq_Def_Complex_Ent},
\begin{align}
   \ent{r}{1-\delta} (\rho)
   \leq  \log_2  \LParen  \Tr ( \bar{Q} )  \RParen
   =  \log_2 ( 2^{\Sites-w} )
   = \Sites-w .
\end{align}
The penultimate equality follows from Eq.~\eqref{eq_Def_Q_0}. 
\end{proof}

\section{Uncomplexity-expenditure proof}
\label{app_Prove_Expend}

This appendix contains the proof of Theorem~\ref{thm_Uncomp_Cost}.
We must upper-bound the cost of simulating a state $\rho$ $\delta$-approximately.
First, we prove Theorem~\ref{thm_Uncomp_Cost} in the absence of fuzziness (Lemma~\ref{lemma_Expend_NoFuzz}).
Then, we use Lemma~\ref{lemma_Triangle_Inequality} to extend the proof to fuzzy gates.

\begin{lemma}   
\label{lemma_Expend_NoFuzz}
Let $\rho$ denote an arbitrary $\Sites$-qubit state.
Let $r$ and $\delta$ be as described above Theorem~\ref{thm_Uncomp_Cost},
but let all gates be implemented perfectly ($\epsilon = 0$).
For every $\delta \in (0, 1]$,
every $\eta \in (0, 1]$, 
and every $(\Sites - w)$-qubit state $\sigma$,
$\rho$ can be imitated with
$w = \Sites - \ent{r}{\eta}(\rho)$ uncomplex $\ket{0}$'s.
\end{lemma}

\begin{proof}
We index the qubits such that 
the referee's measurement operator has the form
\begin{align}
   \label{eq_Expend_Lemma1_Help1}
   \Qref
   = U_r^\dag \left( 
   \ketbra{0^{w'}}{ 0^{w'} }  \otimes
   \id^{\otimes (\Sites - w') }
   \right)U_r \ ,
\end{align}
for some $w' \in \{1,2, \ldots, \Sites \}$ 
and some unitary $U_r$ implementable with $\leq r$ gates.
By the constraints on the referee,
$\Tr ( \Qref \rho ) \geq \eta$.
Hence $\Qref$ satisfies the constraint in
the definition~\ref{def_Comp_Ent} of $\ent{r}{\eta}(\rho)$.

Let the agent request $w = w'$ uncomplex $\ket{0}$'s.
The agent can perform the inverse $U_r^\dag$ of 
the referee's unitary.
The simulacrum acquires the form
\begin{align}
   \label{eq_Expend_Lemma1_Help2}
   \bar{\rho}
   = U_r^\dag  
   \left(  \ketbra{0^w}{0^w}  \otimes  \sigma  \right)
   U_r \ .
\end{align}
If the referee receives $\bar{\rho}$,
their probability of guessing $\rho$ is
$\Tr \left( \Qref \bar{\rho} \right)
= 1  \geq 1-\delta$.
Hence $\bar{\rho}$ satisfies the constraint on the agent.

We can derive two expressions for 
$\log_2 \LParen \Tr (\Qref) \RParen$.
First, by Eq.~\eqref{eq_Expend_Lemma1_Help1} and $w' = w$,
$\log_2 \LParen \Tr (\Qref) \RParen = \Sites - w$.
Second, $\Qref$ was chosen to minimize
$\Tr( \Qref \, \id / 2^\Sites )$.
Therefore, $\Qref$ achieves the minimum in
the definition~\eqref{def_Comp_Ent}.
Therefore, $\log_2 \LParen \Tr (\Qref) \RParen 
= \ent{r}{\eta}(\rho)$.
Equating the two expressions for the log,
and solving for $w$, yields
$w = \Sites - \ent{r}{\eta}(\rho)$.
\end{proof}

We have proved a fuzziness-free variation on Theorem~\ref{thm_Uncomp_Cost}.
We extend that proof to prove the theorem itself,
assuming that all gates applied are $\epsilon$-fuzzy. 

\thmUncompCost*

\begin{proof}
Recall the assumption that $\delta\geq 2r\epsilon$.
Due to gate fuzziness, the referee implements the fuzzy operation
$\tilde{U}_r^{\text{ref}}$, instead of $U_r$, 
and effects the POVM $\tilde{Q}$, instead of $\bar{Q}$. 
Likewise, the agent implements the fuzzy operation
$\tilde{U}_r^{\text{agt}}$, instead of ${U}_r$, 
and constructs the state $\tilde{\rho}$, instead of $\bar{\rho}$. 
The referee identifies $\tilde{\rho}$ as $\rho$ with
probability at least $1-\delta$, since
\begin{align}
    \text{Tr} \left( \tilde{Q}\tilde{\rho} \right)  &=  \text{Tr}\left( \left\{ \tilde{U}_r^{\text{ref}\dagger}
                                                       \left[ \ketbra{0^w}{0^w} \otimes \mathbbm{1}^{\otimes(n-w)} \right]
                                 \tilde{U}_r^{\text{ref}} \right\} \left\{ \tilde{U}_r^{\text{agt}\dagger}
                                                       \left[ \ketbra{0^w}{0^w} \otimes \sigma \right]
                                                       \tilde{U}_r^{\text{agt}} \right\} \right) \\
                                                    &= \text{Tr}\left( \left[ \ketbra{0^w}{0^w} \otimes \mathbbm{1}^{\otimes(n-w)} \right]
                                                    \tilde{U}_r^{\text{ref}} \tilde{U}_r^{\text{agt}\dagger}
                                                    \left[ \ketbra{0^w}{0^w} \otimes \sigma \right] 
                                                    \tilde{U}_r^{\text{agt}} \tilde{U}_r^{\text{ref}\dagger} \right) \\
                                                    &\geq \text{Tr}\left( \left[ \ketbra{0^w}{0^w} \otimes \mathbbm{1}^{\otimes(n-w)} \right]
                                                    \left[ \ketbra{0^w}{0^w} \otimes \sigma \right] \right) \nonumber \\ & \qquad
                                                    \label{Expend_Help9}
                                                    -\mathcal{T}\left( \ketbra{0^w}{0^w} \otimes \sigma, \, \tilde{U}_r^{\text{ref}} \tilde{U}_r^{\text{agt}\dagger}
                                                    \left[ \ketbra{0^w}{0^w} \otimes \sigma \right]
                                                    \tilde{U}_r^{\text{agt}} \tilde{U}_r^{\text{ref}\dagger} \right) \\
                                                    \label{Expend_Help10}
                                                    &= 1 - \mathcal{T}\left( \tilde{U}_r^{\text{ref}\dagger} 
                                                    \left[ \ketbra{0^w}{0^w} \otimes \sigma \right]
                                                    \tilde{U}_r^{\text{ref}}, \, \tilde{U}_r^{\text{agt}\dagger} 
                                                    \left[ \ketbra{0^w}{0^w} \otimes \sigma \right]
                                                    \tilde{U}_r^{\text{agt}} \right) \\
                                                    &\geq 1 - \mathcal{T}\left( \tilde{U}_r^{\text{ref}\dagger}
                                                    \left[ \ketbra{0^w}{0^w} \otimes \sigma \right]
                                                    \tilde{U}_r^{\text{ref}}, \, U_r^\dagger
                                                    \left[ \ketbra{0^w}{0^w} \otimes \sigma \right]
                                                    U_r \right) \nonumber \\ & \qquad
                                                    \label{Expend_Help11} 
                                                    - \mathcal{T}\left( U_r^\dagger 
                                                    \left[ \ketbra{0^w}{0^w} \otimes \sigma \right]
                                                    U_r, \, \tilde{U}_r^{\text{agt}\dagger} 
                                                    \left[ \ketbra{0^w}{0^w} \otimes \sigma \right]
                                                    \tilde{U}_r^{\text{agt}} \right) \\
                                                    \label{Expend_Help12}
                                                    &\geq 1 - 2r\epsilon \\
                                                    &\geq 1 - \delta.
\end{align}
Inequality~\eqref{Expend_Help9} follows by the application of Ineq.~\eqref{eq_Wilde_Ineq} to $\sigma = \tilde{U}_r^{\text{ref}} \tilde{U}_r^{\text{agt}\dagger} \left[ \ketbra{0^w}{0^w} \otimes \sigma \right] \tilde{U}_r^{\text{agt}} \tilde{U}_r^{\text{ref}\dagger}$, $\gamma = \ketbra{0^w}{0^w} \otimes \sigma$, and $\Lambda = \ketbra{0^w}{0^w} \otimes \mathbbm{1}^{\otimes(n-w)}$. Equation~\eqref{Expend_Help10} follows from the trace distance's unitary invariance. Inequality~\eqref{Expend_Help11} follows from the triangle inequality. Inequality~\eqref{Expend_Help12} follows from Lemma~\ref{lemma_Triangle_Inequality}.
Therefore, the agent can imitate $\rho$ with probability $\geq 1 - \delta$, using (as shown in Lemma~\ref{lemma_Expend_NoFuzz})
$w = \Sites - \ent{r}{\delta} (\rho)$
uncomplex $\ket{0}$'s.
\end{proof}

\section{Monotonicity of the complexity negentropy in two cases}
\label{app_H_Monoton}

This section supports our conjecture that the complexity negentropy declines monotonically under fuzzy operations (Conjecture~\ref{conj_H_Mon}).
We show that the complexity entropy grows monotonically in two cases.
In both, the error-intolerance parameter $\eta = 1$.
The reason is, our techniques derive from a proof of the linear growth, under random circuits, of the \emph{exact complexity}, the least number of two-qubit gates required to prepare an $\Sites$-qubit state $\ket{\psi}$ from $\ket{0^\Sites}$ exactly~\cite{Haferkamp_21_Linear}.
Define as the \emph{approximate complexity} the least number of two-qubit gates required to prepare $\ket{\psi}$ approximately.
We would have to prove the approximate complexity's linear growth to prove that $\ent{r}{\eta}$ increases monotonically under fuzzy operations when $\eta < 1$.
The approximate proof would require more than the dimension counting used here; we would need insights into the geometry of the set of unitaries implemented by local quantum circuits.

Our proof involves two definitions, which we repeat from the main text:
(i) Define an \emph{architecture} as the layout of gates in a quantum circuit.
(ii) Define as $\mathcal{E}_{k,r}$ the ensemble of 
$\Sites$-qubit states formed as follows:
Pick $k = 0, 1, \ldots, \Sites$ qubits uniformly randomly;
and pick a $k$-qubit state vector $\ket{\phi}$ Haar-randomly.
Prepare those qubits in $\ket{\phi}$.
Pad $\ket{\phi}$ with $\ket{0}$'s, to produce 
$\ket{\phi}  \ket{ 0^{\Sites - k} }$.
Perform a circuit, with a random architecture, of any $\leq r$ fuzzy gates~\footnote{
We draw the random architecture, or directed graph, as follows:
Sequentially draw $R$ uniformly random pairs of qubits $(j, k)$, with $j\neq k$.
For each pair, add to the graph a vertex between edges $j$ and $k$.}.

\monotoneentropyI*

\begin{proof}
Consider the ``uncomplex'' POVM elements $Q \in M_r$ for which
$\log_2 \LParen \mathrm{Tr}(Q) \RParen \leq k$.
Consider the states $\ket{\psi}$ for which,
for some such $Q$,
$\mathrm{Tr}(Q \ketbra{\psi}{\psi}) 
   = 1.$
All the $Q \in M_r$ are mutually orthogonal projectors.
Therefore, the aforementioned $\ket{\psi}$'s form the set
\begin{equation}
   \label{eq_N_r_k}
   \mathcal{U}_{r,k}
   := \bigcup_{\substack{Q\in M_r : \\  
               \log_2 \LParen \Tr(Q) \RParen \leq k}} 
   \mathrm{Im}(Q)
   = \bigcup_{\substack{Q\in M_r : \\  
               \log_2 \LParen \Tr(Q) \RParen \leq k}} 
   \bigcup_{\pi\in \mathrm{S}_n} 
   \Set{U_r U_{r-1} \ldots U_1 \,
   \pi \, \ket{\phi} |0^{n-k} \rangle \, : \,
   U_j
   \text{ is two-local } \forall j, \;
   |\phi \rangle\in (\mathbb{C}^{2})^{\otimes k }}.
\end{equation}
The projector $Q$ has an image Im$(Q)$, and
S$_\Sites$ denotes the group of the permutations of $\Sites$ objects.
Since $\ket{0^k}$ is a state $|\phi \rangle\in \mathbb{C}^{k}$,
$\mathcal{U}_{r,0} \subseteq \mathcal{U}_{r,1}\subseteq \ldots \subseteq \mathcal{U}_{r,k}.$

To characterize the $\mathcal{U}$'s further, 
we denote by $S^{D}$ the sphere in $\mathbb{R}^{D+1}$ \footnote{Our proof could be cast equivalently in terms of complex projective spaces, which are prevalent in the holographic literature \cite{Brown_18_Second}.}.
Also, we identify $\mathbb{C}^{2^k}$ with $\mathbb{R}^{2 \times 2^k}$.
The $\mathcal{U}_{r,k}$ are the images of 
the polynomial contraction maps 
$S^{2 \times 2^k-1}\times \mathrm{SU}(4)^{\times R}
\to (\mathbb{C}^2)^{\otimes k}$.
Therefore, by the Tarski-Seidenberg principle~\cite{Bochnak_13_Real}, the $\mathcal{U}_{r,k}$ are semialgebraic sets. (A semialgebraic set consists of the solutions to a finite set of polynomial equations and inequalities over the real numbers.)
Every semialgebraic set decomposes as a union of smooth manifolds~\cite{Bochnak_13_Real}.
The greatest manifold dimension is the semialgebraic set's dimension~\cite{Bochnak_13_Real}.
Therefore, we can bound $\dim (\mathcal{U}_{r,k})$ as follows.
Since $\{|\phi\rangle |0^{n-k}\rangle\}\subseteq \mathcal{U}_{r,k}$, 
$\dim (\mathcal{U}_{r,k}) 
\geq \dim \left( S^{2 \times 2^k-1} \right) 
= 2 \times 2^k-1$.
According to Lemma~\ref{lemma:monotoneentropyI},
$k>\log_2(15r)$.
Therefore, $2 \times 2^k-1>2 \times 2^{k-1}-1+15r$. 
By parameter counting, $\dim (\mathcal{U}_{r,k-1})\leq 2 \times 2^{k-1}-1+15r$. Therefore,
$\dim (\mathcal{U}_{r,k-1})
   < \dim (\mathcal{U}_{r,k})$.
Therefore, by Ref.~\cite[Lemma~1]{Haferkamp_21_Linear}, $\mathcal{U}_{r,k-1}$ forms a measure-0 set in $\mathcal{U}_{r,k}$.

We apply the above conclusion as follows.
Consider drawing a state uniformly randomly from the ensemble $\mathcal{E}_{k,r}$.
If the state has a complexity entropy $\ent{r}{1} \leq k-1$,
the state is in $\mathcal{U}_{r,k-1}$,
which forms a measure-0 set in $\mathcal{U}_{r,k}$,
we just concluded.
Therefore, the drawn state satisfies $\ent{r}{1} \leq k-1$ with probability 0. 
Every unitary acts on $\mathcal{U}_{r,k}$ as a diffeomorphism
and so does not change the set's dimension.
(The unitary maps $\mathcal{U}_{r,k}$ to a set of ``just as many'' unitaries.)
Fuzzy operations are unitaries, so applying a fuzzy operation cannot decrease $\ent{r}{1}$.
\end{proof}

Having proved that the complexity negentropy decreases monotonically in one case, we proceed to the second case. 
Denote by $\mathcal{E}_{k,A}$ the ensemble defined as $\mathcal{E}_{k,r}$, except that the $r$-gate fuzzy circuit 
has the architecture $A$. 

\begin{lemma}[Second case of the complexity negentropy's monotonicity] \label{lemma:monotoneII}
For every $k<n$ and $r = 0, 1, \ldots, 2 \times \lfloor(2^n-1-2^k)/15\rfloor$, 
there exists an architecture $A$ for which the following holds:
Let $A'$ denote any architecture that contains a light cone.
Consider drawing
(i) a state from $\mathcal{E}_{k,A}$ and 
(ii) an architecture-$A'$ circuit.
Consider following $A$ with the light-cone--containing $A'$, and call the extended architecture $A_{\mathrm{ext}}$.
Running the circuit on the state decreases the complexity negentropy
$\Sites - \ent{r}{1}$, with probability 1 over the ensemble $\mathcal{E}_{k,A_{\mathrm{ext}}}$.
\end{lemma}

\begin{proof}
Define, similarly to Eq.~\eqref{eq_N_r_k}, the set
\begin{equation}
   \mathcal{U}_{k,A} 
   := \bigcup_{\substack{Q\in M_r : \\  
               \log_2 \LParen \Tr(Q) \RParen \leq k}} 
   \left\{ U_r U_{r-1} \ldots U_1 
   |\phi\rangle \ket{0^{n-k}} \, : \, 
\text{circuit has architecture}~A, \; 
|\phi\rangle\in (\mathbb{C}^{2})^{\otimes k }  \right\}.
\end{equation}
We proceed similarly to the proof of Theorem~\ref{thm_Monotones}:
Let the architecture $A$ be such that 
$\dim (\mathcal{U}_{k,A})$ has the greatest value achievable with any $(\leq r)$-gate architecture.
We must prove only that 
$\dim (\mathcal{U}_{k,A}) < \dim (\mathcal{U}_{k,A_{\mathrm{ext}}})$:
If this inequality holds, then, by Lemma~\ref{lemma:subsetlemma},
randomly drawing a state from 
$\mathcal{E}_{k,A_{\mathrm{ext}}}$
has 0 probability of being in $\mathcal{U}_{k, \tilde{A}}$, 
for every architecture $\tilde{A}$ with $\leq r$ gates.

The dimension $\dim (\mathcal{U}_{k,A})$ equals
the dimension of a contraction map's Jacobian.
The Jacobian's image is spanned by a set of vectors 
$|v_j\rangle \in \mathbb{R}^{2 \times 2^\Sites-1}$.
We can choose for the vectors to have the form
$|v_j\rangle
= A_j \, U_1 U_2 \ldots U_r|0^{\Sites-k}\rangle |\phi\rangle ,$
for Hermitian operators $A_j$.
There is a Pauli operator $P$ such that 
$|v'\rangle = P U_1 U_2 \ldots U_r|0^{\Sites-k}\rangle 
|\phi\rangle\notin \mathrm{span}\{|v_j\rangle\}$, 
if the Jacobian's image has a rank less than the greatest value possible, $2 \times 2^n-1$ \footnote{
This claim follows as in the proof of Theorem~\ref{thm_Monotones}.}.
The rank's submaximality is guaranteed by the assumptions $k<n$ and 
$r = 0, 1, \ldots, 2 \times \lfloor (2^n-1-2^k)/15\rfloor$.
We can follow $A$ with an architecture-$A'$, depth-$R'$ circuit.
Applying the procedure of Ref.~\cite{Haferkamp_21_Linear} to the Pauli operator $P$, we find a higher-rank point for $A_{\mathrm{ext}}$:
$\dim (\mathcal{U}_{k,A}) < \dim (\mathcal{U}_{k,A_{\mathrm{ext}}}).$
\end{proof}

\end{appendices}

%
%
\bibliographystyle{h-physrev}
\bibliography{Uncomplexity_Bib}

\end{document}